\documentclass[12pt]{amsart}
\usepackage{amsmath,amsfonts,amssymb,bbm}
\usepackage[numbers,sort&compress]{natbib} 
\usepackage{bm}

\usepackage{booktabs} 
\usepackage[margin = 1in]{geometry}
\usepackage[utf8]{inputenc}
\usepackage{ifthen}
\usepackage[boxed, linesnumbered]{algorithm2e}
\usepackage{comment}
\usepackage{graphicx, caption}
\usepackage{pgf}
\usepackage{tikz}
\usetikzlibrary{arrows,automata}
\usepackage{quote}

\usepackage{amsthm}

\usepackage[colorlinks=true,linkcolor=blue]{hyperref}

\newtheorem{theorem}{Theorem}[section]
\newtheorem{corollary}[theorem]{Corollary}
\newtheorem{lemma}[theorem]{Lemma}

\newtheorem{definition}[theorem]{Definition}
\newtheorem{remark}[theorem]{Remark}
\newtheorem{proposition}[theorem]{Proposition}

\newcommand{\eps}{\mbox{$\epsilon$}}

\newcommand{\defeq}{\stackrel{\textup{def}}{=}}

\renewcommand{\vec}[1]{\mathbf{#1}}

\def\E{\ensuremath{\mathbf{E}}}

\SetArgSty{textrm}  
\SetAlFnt{\small}
\SetAlCapFnt{\small}
\SetAlCapNameFnt{\small}
\SetAlCapHSkip{0pt}
\IncMargin{-\parindent}

\begin{document}

\title[The route to chaos  in routing games 
]{
{The route to chaos  in routing games:\\}  
{
When is Price of Anarchy too optimistic?}
}

\author[T. Chotibut]{Thiparat Chotibut}
\address[T. Chotibut]{Department of Physics, Faculty of Science, Chulalongkorn University, Bangkok 10330, Thailand \newline 
 Engineering Systems and Design, Singapore
  University of Technology and Design, 8 Somapah Road, Singapore 487372 }
\email{Thiparat.C@chula.ac.th, thiparatc@gmail.com}

\author[F. Falniowski]{Fryderyk Falniowski}
\address[F. Falniowski]{Department of Mathematics, Cracow University
  of Economics, Ra\-ko\-wicka~27, 31-510 Krak\'ow, Poland}
\email{fryderyk.falniowski@uek.krakow.pl}

\author[M. Misiurewicz]{Micha{\l} Misiurewicz}
\address[M. Misiurewicz]{Department of Mathematical Sciences, Indiana
  University-Purdue University Indianapolis, 402 N. Blackford
  Street, Indianapolis, IN 46202, USA}
\email{mmisiure@math.iupui.edu}

\author[G. Piliouras]{Georgios Piliouras}
\address[G. Piliouras]{Engineering Systems and Design, Singapore
  University of Technology and Design, 8 Somapah Road, Singapore 487372}
\email{georgios@sutd.edu.sg}

\begin{abstract}
Routing games are amongst the most studied classes of games. Their two most well-known properties are that learning dynamics converge to equilibria and that all equilibria are approximately optimal. In this work, we perform a stress test for these classic results by studying the ubiquitous dynamics, Multiplicative Weights Update, in different classes of congestion games, uncovering intricate non-equilibrium phenomena. \textit{As the system demand increases}, the learning dynamics go through period-doubling bifurcations, leading to instabilities, chaos and large inefficiencies even in the simplest case of non-atomic routing games with two paths of linear cost where the Price of Anarchy is equal to one.

Starting with this simple class, we show that every system has a carrying capacity, above which it becomes unstable. If the equilibrium flow is a symmetric $50-50\%$ split, the system exhibits one period-doubling bifurcation. A single periodic attractor of period two replaces the attracting fixed point. Although the Price of Anarchy is equal to one,  in the large population limit  the time-average social cost for all but a zero measure set of initial conditions converges to its worst possible value. For asymmetric equilibrium flows, increasing the demand  eventually forces the system into Li-Yorke chaos with positive topological entropy and periodic orbits of all possible periods.  Remarkably, in all non-equilibrating regimes, the time-average flows on the paths converge {\it exactly} to the equilibrium flows, a property akin to no-regret learning in zero-sum games. These results are \textit{robust}.  We extend them to routing games with arbitrarily many strategies, polynomial cost functions, non-atomic as well as atomic routing games and heteregenous users. Our results are also applicable to any sequence of shrinking learning rates, e.g., $1/\sqrt{T}$, by allowing for a dynamically increasing population size.

\end{abstract}


\maketitle

\vspace{-32pt} 

\begin{figure}[ht]

\centering
\includegraphics[width=0.67\textwidth]{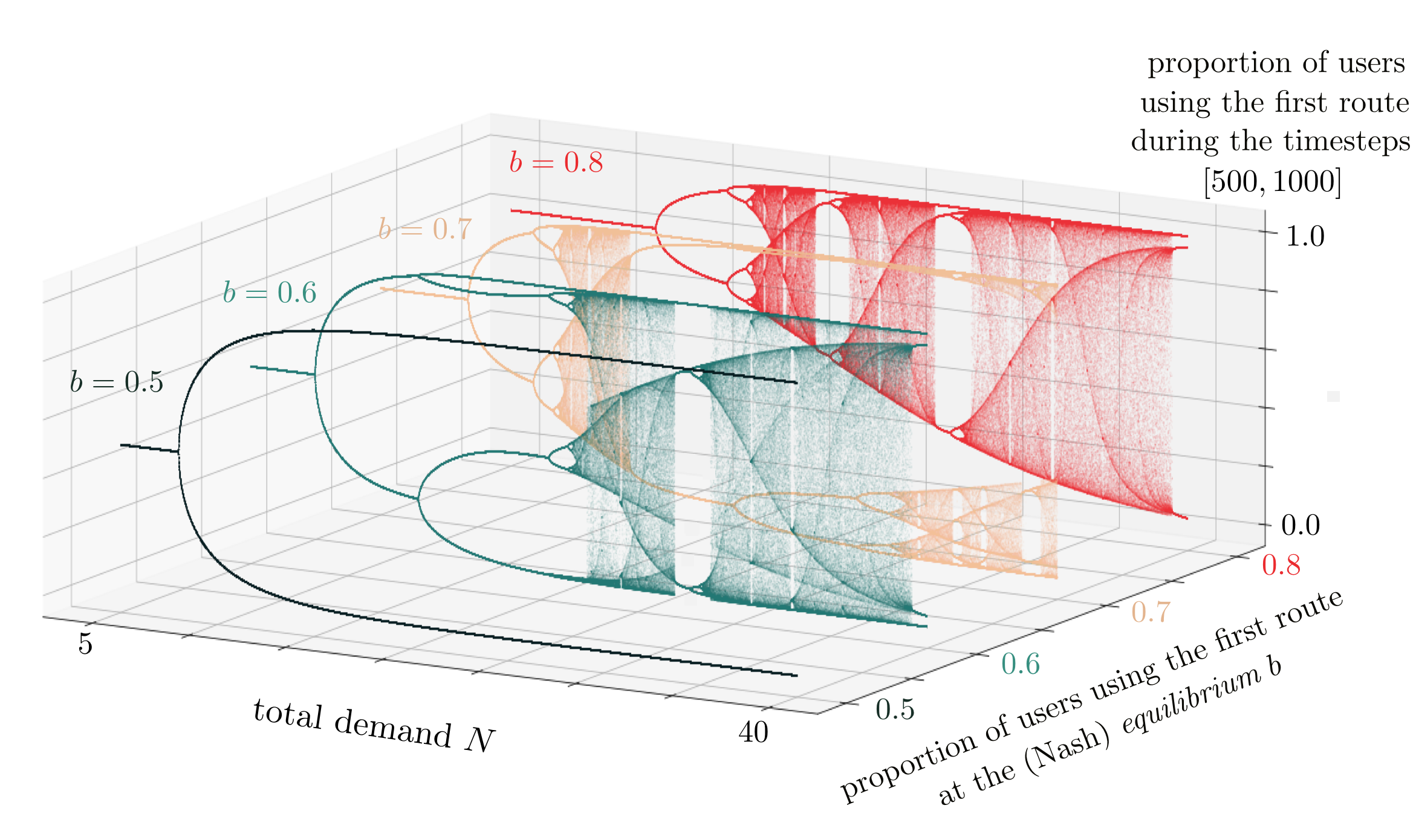} 
\caption{If the equilibrium flow is symmetric between the two routes ($b = 0.5$), large total demand $N$ leads to a limit cycle of period $2$. In any game with an asymmetric equilibrium split ($b\neq 0.5$), chaos emerges at large $N$. For detailed discussions, see Figure \ref{fig: intro_summary}.}
\label{fig: intro_summary3D}

\end{figure}


\newpage

\section{Introduction}
\label{s:intro}

Congestion and  routing games \cite{rosenthal73} are amongst the most well studied class of games in game theory.  Being isomorphic to potential games \cite{monderer1996fictitious}, congestion games are one of the few classes of games in which a variety of learning dynamics are known to converge to Nash equilibria  \cite{Even-Dar:2005:FCS:1070432.1070541,Fotakis08,Fischer:2006:FCW:1132516.1132608,Kleinberg09multiplicativeupdates,kleinberg2011load,hoo}. Proving convergence to equilibria typically exploits the existence of the potential function that acts as a (strong) Lyapunov function for learning dynamics; this function is strictly decreasing when the system is out-of-equilibrium.

Congestion games also play a pivotal role in the study of Price of Anarchy \cite{KoutsoupiasP99WorstCE,roughgarden2002bad,christodoulou,Fotakis2005226,bilo2017lookahead,correa2018inefficiency}. Price of Anarchy (PoA)  is defined as the ratio of the social cost of the worst Nash equilibrium to the optimal social cost. 
 A small Price of Anarchy implies that all Nash equilibria are near optimal, and hence any equilibrating learning dynamics suffices to reach approximately optimal system performance. One of the hallmarks of the Price of Anarchy research has been the development of tight Price of Anarchy bounds for congestion games that are independent of the topology of the network or the number of users. Specifically, under the prototypical assumption of linear cost functions, Price of Anarchy in the case of non-atomic agents (in which each agent controls an infinitesimal amount of flow) is at most $4/3$ \cite{roughgarden2002bad}. In the atomic case (in which each agent controls a discrete unit of flow), Price of Anarchy is at most $5/2$ \cite{christodoulou}, with small networks sufficing to provide tight lower bounds.

Additionally, congestion games have paved the way for recent developments in  Price of Anarchy research, extending our understanding of system performance even for non-equilibrating dynamics.
 Roughgarden \cite{Roughgarden09} showed that most Price of Anarchy results could be organized in a common framework known as $(\lambda,\mu)$-smoothness.
For classes of games that satisfy this property, such as congestion games, the Price of Anarchy bounds derived for worst case Nash equilibria immediately carry over to worst case instantiations of regret minimizing algorithms.
 An (online) algorithm is said to minimize regret as long as its time-average performance is roughly as good as that of the best fixed action with hindsight. The most ubiquitous member of this class of algorithms is arguably the Multiplicative Weights Update (MWU) \cite{Arora05themultiplicative}.  The aforementioned Price of Anarchy results readily apply to any learning dynamics or any sequence of strategic plays, as long as the algorithms achieve small time-average regret. In the case of congestion games, near optimal time-average performance is guaranteed \textit{asymptotically} even when learning does not equilibrate, however, as we discuss in Section \ref{s:discussion} 
  slow convergence rates that reduce the applicability of such results in some applications of interest, including congestion games with many agents. 

All these positive results have inspired growing efforts to achieve  stronger efficiency guarantees. How close to $1$ can the Price of Anarchy get? For example, in linear congestion games, what is the final correct answer after all? Is it $5/2$ as the atomic model suggests, or is the much better non-atomic bound of $4/3$ binding? In the case of polynomial cost functions, the gap between these predictions grows exponentially fast with the degree of the polynomials, making this question even more pressing.
In a recent development, \cite{Feldman:2016:PAL:2897518.2897580} argues that even if the agents are atomic, as long as the number of agents $N$ is large, atomic congestion games behave approximately as non-atomic ones; hence, $4/3$ is the correct bound. That is, the much smaller non-atomic bound is the correct one after all.
A series of related results has followed \cite{colini2016price, colini2017asymptotic,colini2018price}, suggesting strong bounds on Price of Anarchy under the assumption of large demand. What is the connection between atomic and non-atomic congestion games?

Let us consider the simplest congestion game example. 
A game with two strategies and two agents where the cost/latency of each strategy is equal to its load. The worst Nash equilibrium of this game has both agents choosing a strategy uniformly at random. The expected cost of each agent is $3/2$, i.e., a cost  equal to $1$ due to their own load, and an expected extra cost of $1/2$ due to the $50\%$ chance of adopting the same strategy as the other agent. On the other hand, at the optimal state, each agent selects a distinct strategy at a cost of $1$. As a result, the Price of Anarchy for this game is $3/2$.
Suppose now that we increase the number of agents from $2$ to $N$\footnote{For simplicity, let $N$ be an even number.}. The worst equilibrium still has each agent choosing a strategy uniformly at random at an expected cost of $(N-1)/2+1=(N+1)/2$. The optimal configuration splits the agents, deterministically and equally to both strategies at a cost of $N/2$ per agent. The Price of Anarchy is $1+1/N,$ converging to $1$ as $N$ grows. Indeed, as the population size grows, the atomic game is more conveniently described by its effective non-atomic counterpart, with a continuum of users, and a unique equilibrium that equidistributes the total demand $N$  between the two strategies. 
So, the equilibria indeed are effectively optimal. How does this large demand, however, affect the dynamics?

 {\bf Informal Meta-Theorem:} We analyze MWU in routing/congestion games under a wide range of settings and combinations thereof (two/many paths, non-atomic/atomic, linear/polynomial, etc).  Given any such game $G$ and an arbitrary small learning rate (step-size) $\epsilon$, we show that 
 there exist a system capacity $N_0(G, \epsilon)$ such that \textit{if  
 the total demand exceeds this threshold the system is provably unstable} with complex non-equilibrating behavior. Both \textit{periodic behavior} as well as  \textit{chaotic behavior} is proven and we give formal guarantees about the conditions under which they emerge. Despite this unpredictability of the non-equilibrating regimes the \textit{time-average} costs/flows 
 exhibit regularity and for linear costs
 \textit{converge to equilibrium}. The variance, however, of the resulting non-equilibrium flows leads to increased inefficiencies, as the \textit{time-average cost can be  arbitrarily high}, even for simple games where all equilibria are optimal.  
 

 {\bf Intuition behind instability:}
 To build an intuition of why instability can arise, let us revisit the simple example with two strategies and let us consider a continuum/large number of users 
  updating their strategies according to a learning dynamic, e.g. MWU with a step-size $\epsilon$. Given any non-equilibrium initial condition, the agents on the over-congested strategy have a strong incentive to migrate to the other strategy. As they all act in unison, \textit{if the total demand is sufficiently large}, the corrective deviation to the other strategy will be overly aggressive, resulting in the other strategy becoming over-congested. 
With this heuristic consideration, a self-sustaining non-equilibrating behavior, where users bear higher time-average costs than those at the equilibrium flow, is indeed plausible. In this work, we show that it is in fact provably true, even for games with arbitrarily many strategies.

 {\bf Significance of results:} \textit{When is (robust) PoA analysis too optimistic?} 
 In our simple example with two parallel links and a continuum/large number of users we have that the PoA of the game is equal to $1$. All linear non-atomic congestion game have an upper bound PoA of $4/3$ and in the case of atomic congestion games with many agents, recent work \cite{Feldman:2016:PAL:2897518.2897580} implies a  $(\lambda,\mu)$-smoothness robust PoA bound of the same magnitude. However, for any arbitrarily small fixed $\epsilon$ (e.g. $\epsilon=2^{-100}$), we can choose a large enough demand/population size such that the time average cost is the worst possible, which is a factor of $2$ away from optimal (Theorems \ref{t:SC} and \ref{t:Disc1}). Hence, robust PoA bounds do not necessarily reflect an upper bound on the actual behavior of MWU under arbitrarily small fixed learning rates.  Applying \textit{a decreasing learning rate $\epsilon$ is not an easy fix} either. 
 Our work reveals that there exists an effective tension between the size of the step-size/learning rate $\epsilon$ of the dynamic and the total system demand/population size $N$.
 Despite the  undeniable usefulness of PoA analysis and $(\lambda,\mu)$-smoothness robust PoA, 
 in games with many agents even with decreasing step-sizes 
 these Price of Anarchy guarantees do not become binding for regret-minimizing dynamics until after arbitrarily long chaotic histories of possibly arbitrarily high time-average social cost. Our analysis peers exactly into these non-equilibrating regimes and by understanding their geometry and variance show that they can indeed imply worst case social costs. 
  Moreover, if we incorporate a slowly increasing population size then even with decreasing step sizes the system will provable stay in its chaotic, inefficient regime forever.
  See Section \ref{s:discussion} for a detailed discussion of these issues.

%


{\bf Base model \& results:} We start by focusing on the minimal case of linear non-atomic congestion games with two edges and total demand $N$. All agents are assumed to evolve their behavior using Multiplicative Weights Updates with an arbitrarily small, fixed learning rate $\epsilon$.
In Section \ref{s:1d} we prove that every such system has a critical threshold, a hidden system capacity, which when exceeded, the system exhibits a bifurcation and no longer converges to its equilibrium flow.
 If the unique equilibrium flow is the $50-50\%$ split (doubly symmetric game),  the system proceeds through exactly one period-doubling bifurcation, where a single attracting periodic orbit of period two replaces the attracting fixed point. In the case where the game possesses an asymmetric equilibrium flow, the bifurcation diagram is much more complex (Figures \ref{fig: intro_summary3D}, \ref{fig: intro_summary} and \ref{fig: cobwebb_b0p7}). As the total demand changes, we will see the birth and death of periodic attractors of various periods. All such systems provably exhibit Li-Yorke chaos, given sufficiently large total demand.  This implies that there exists an uncountable set of initial conditions such that the set is "scrambled", i.e.,  given any two initial conditions $x(0), y(0)$ in this set,  $\liminf dist(x(t),y(t))=0$ while $\limsup dist(x(t),y(t))>0$. 
 Everywhere in the non-equilibrating regime, MWU's time-average behavior is reminiscent of its behavior in zero-sum games. Namely, the time-average flows and costs of the strategies converge {\it exactly} to their \textit{equilibrium values}. Unlike zero-sum games, however, these non-equilibrium dynamics exhibit large regret (Section \ref{s: regret}), and (possibly arbitrarily) high time-average social costs (Section \ref{s: socialcost}), even when the Price of Anarchy is equal to one. In Section \ref{s:entropy}, we argue that the system displays another signature of chaotic behavior, positive topological entropy.  We provide an intuitive explanation by showing that 
 if we encode three events: A) the system is approximately at equilibrium, B) the first strategy is overly congested, C) the second strategy is overly congested, then the number of possible sequences on the alphabet $\{A,B,C\}$,  encapsulating possible system dynamics, grows exponentially with time. Clearly, if the system reached an (approximate) equilibrium, any sequences must terminate with an infinite string of the form $\dots$AAA$\dots$. Instead, we find that the system can become truly unpredictable.
 In Appendix \ref{s:properties_orbits}, we show that the system may possess multiple distinct attractors and hence the time-average regret and social cost depend critically on initial conditions. Properties of periodic orbits, the evidence of Feigenbaum's universal route to chaos in our non-unimodal map, are also provided.

{\bf Extensions:} We conclude the paper by examining the robustness of our findings. 
In Appendix \ref{s:extension}, we prove that our results hold not only for graphs with two paths but extend for arbitrary number of paths.
In Appendix \ref{s:nonlinear} ,we prove Li-Yorke chaos, positive topological entropy and time-average results for 
polynomial costs.
In Appendix \ref{s:mixture}, we provide extensions for games with heterogeneous users.
Finally, in Appendix \ref{s:reductions}, we produce a reduction for MWU dynamics from atomic to non-atomic   congestion games.
This allows us to extend our proofs of chaos, inefficiency to atomic congestion games with many agents.




\begin{figure}[h!]
\centering
\includegraphics[width=0.68\textwidth]{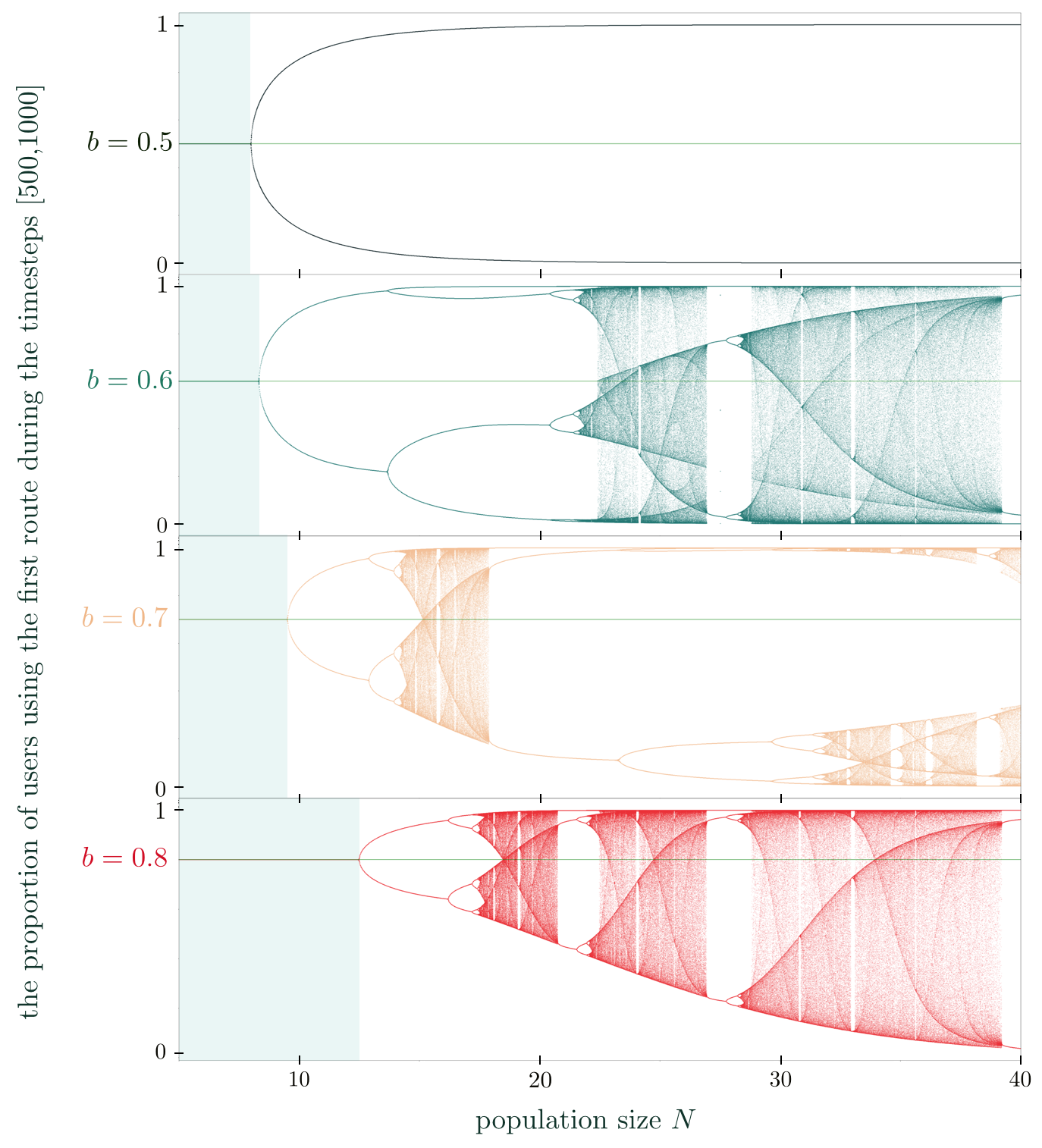}
\caption{These bifurcation diagrams summarize the non-equilibrium phenomena identified in this work. When Multiplicative Weights Update (MWU) learning is applied on a non-atomic linear congestion games with two routes, population increase drives period-doubling instability and chaos. Standard equilibrium analysis only holds at small population sizes, shown in light cyan regions. As population size $N$ (up to a rescaling factor of a fixed learning rate) increases, regret-minimizing MWU algorithm {\it no longer converges} to the Nash equilibrium flow $b$, depicted as the green horizontal lines; the proportion of users using the first route deviates significantly from the Nash equilibrium flow. When the equilibrium flow is symmetric between the two routes ($b = 0.5$), large $N$ leads to non-equilibrium dynamics that is attracted toward a limit cycle of period two. For large $N$, the two periodic points can approach $1$ or $0$ arbitrarily close, meaning that almost all users will occupy the same route, while simultaneously alternating between the two routes. Thus, the time-average social cost can become as bad as possible. In {\it any} game with an asymmetric equilibrium flow ($b \neq 0.5$), Li-Yorke chaos is inevitable as $N$ increases. Although the dynamics is non-equilibrating or chaotic, the time-average of the orbits still converges {\it exactly} to the equilibrium $b$. 
 This work proves the aforementioned statements, and investigates the implications of non-equilibrium dynamics on the standard Price of Anarchy analysis. Properties of chaotic attractors and of the period-doubling bifurcations, as well as extensions to more complex congestion games are studied in the appendices.}
\label{fig: intro_summary}
\end{figure}


\section{Model}
\label{s:prelim}
We consider a two-strategy \emph{congestion game} (see~\cite{rosenthal73}) with a continuum of players(agents), where all of them apply the \emph{multiplicative weights update} to update their strategies \cite{Arora05themultiplicative}. Each of the players controls an infinitesimal small fraction of the flow.
  We will assume that the total flow of all the players is equal to $N$. 
  We will denote the fraction of the players adopting the first strategy at time $n$ as $x_n$. The second strategy is then chosen by $1-x_n$ fraction of the players. This model physically encapsulates how a large population of commuters selects between the two alternative paths that connect the initial point to the end point.
When a large fraction of the players adopt the same
strategy, congestion arises, and the cost of choosing the same strategy increases.

{\bf Linear routing games}:
We 
focus on linear cost functions. Specifically, the cost of each path (link, route, or strategy) here will be assumed proportional to the \emph{load}. By denoting $c(j)$ the cost of selecting the
strategy number $j$ (when $x$ fraction of the agents choose the first strategy), if the coefficients of proportionality are
$\alpha,\beta>0$, we obtain
\begin{equation}\label{cost}
c(1)=\alpha N x, \hspace{50pt} c(2)=\beta N (1-x).
\end{equation}

Our analysis on the emergence of bifurcations, limit cycles and chaos will carry over immediately to the cost functions of the form $\alpha x+\gamma$. As we will see, the only parameter that is important is the value of the equilibrium split, i.e. the percentage of players using the first strategy at equilibrium. The first advantage of this formulation is that the fraction of agents using each strategy at equilibrium is independent of the flow $N$.
The second advantage is that the Price of Anarchy of these games is exactly $1$, independent of $\alpha, \beta, $ and $N$. Hence, our model offers a natural benchmark for comparing equilibrium analysis, which suggests optimal social cost, to the time-average social cost arising from non-equilibrium learning dynamics which as we show can be as large as possible.

\subsection{Learning in congestion games with multiplicative weights}

At time $n+1$, we assume the players know the cost of the strategies at
time $n$ (equivalently, the probabilities $x_n, 1-x_n$) and update
their choices. Since we have a continuum of agents, 
the realized flow (split) is accurately described by the probabilities $(x_n, 1-x_n)$.
The algorithm for updating the probabilities that we focus on is the
\emph{multiplicative weights update} (MWU), the ubiquitous learning algorithm widely employed in Machine Learning, optimization, and game theory \cite{Arora05themultiplicative,Nisan:2007:AGT:1296179,roughgarden2016twenty}. Namely, there is a
parameter $\eps\in(0,1)$, which can be treated as the  common learning rate of
all players, such that each probability gets multiplied
by $(1-\eps)$ to the power which is the cost of playing a given
strategy by the given player. 
 The numbers obtained in this way usually
will not be probabilities, so we have to normalize them. Thus, we get
\begin{equation}\label{mwu}
\begin{aligned}
x_{n+1}&=\frac{x_n(1-\eps)^{c(1)}}{x_n(1-\eps)^{c(1)}+
(1-x_n)(1-\eps)^{c(2)}}\\
           &=\frac{x_n}{x_n+
(1-x_n)(1-\eps)^{c(2)-c(1)}}.\\
\end{aligned}
\end{equation}
In this way, a large cost at time $n$ will decrease
the probability of choosing the same strategy at time $n+1$.

By substituting into~\eqref{mwu} the values of the cost functions
from~\eqref{cost} we get:
\begin{equation}\label{mwu1}
\begin{aligned}
x_{n+1}&=\frac{x_n(1-\eps)^{\alpha N x_n}}{x_n(1-\eps)^{\alpha N x_n}+
(1-x_n)(1-\eps)^{\beta N (1-x_n)}}\\
&=\frac{x_n}{x_n+(1-x_n)(1-\eps)^{\beta N -(\alpha+\beta) N x_n}}.
\end{aligned}
\end{equation}

We introduce the new variables
\begin{equation}\label{var}
a=(\alpha+\beta) N \ln\left(\frac1{1-\eps}\right),\ \ \ b=\frac{\beta}{\alpha+\beta}.
\end{equation}

{In fact, we can assume without loss of generality that $\alpha+\beta=1$ (i.e., by the transformation $\alpha'= \frac{\alpha}{\alpha+\beta}$, $\beta'= \frac{\beta}{\alpha+\beta}$, and $\epsilon'=1- (1-\epsilon)^{\alpha+\beta}$). Under these assumptions, equations (\ref{var}) simplify to 
\begin{equation}\label{var2}
a=N \ln\left(\frac1{1-\eps}\right),\ \ \ b=\beta.
\end{equation}}

We will thus study the dynamical systems generated by the one-dimensional map:
\begin{equation}\label{map}
\begin{aligned}
f_{a,b}(x)&=\frac{x}{x+(1-x)\exp(a(x-b))}.\\
\end{aligned}
\end{equation}

As commonly adopted as a standard assumption, the learning rate $\epsilon$ can be regarded as a small, fixed constant in the following analysis but the exact value of constant $\epsilon$ is not of particular interest as our analysis/results will hold for any fixed choice of $\epsilon$ no matter how small. Setting $\epsilon=1-1/e$ such that $\ln\left(\frac1{1-\eps}\right)=1$ simplifies notation as under this assumption $a=N$.
We will then study the effects of the remaining two parameters on system performance, i.e.  $a$, the \textit{(normalized) system demand} and $b$, the \textit{(normalized) equilibrium flow}. 
When $b=0.5$ the routing game is fully symmetric; whereas, when $b$ is close to $0$ or $1$, the routing instance becomes close to a Pigou network with almost all agents selecting the same edge at equilibrium.

\subsection{Regret, Price of Anarchy and time average social cost}\label{sec: regretPoASC}

We will now consider this game from the perspective of each agent as an instance of an online optimization problem.
Consider the set $A=\{1,2\}$ of $2$ actions and a time horizon $T \ge 1.$  At each time step $n = 1,2,\dots,T:$ A decision maker picks a probability distribution ${\bf x}_n= (x_n, 1-x_n)$ over her actions $A$. An adversary picks a cost vector ${\bf c}_n: A \rightarrow [-1,1]$. An action $a_n$ is chosen according to the distribution ${\bf x}_n$, and the decision-maker receives reward $r_n(a_n)$. The decision-maker learns ${\bf r}_n$, the entire reward vector.

An {\it online decision-making algorithm} such as MWU specifies for each $n$ the probability distribution ${\bf x}_n $, as a function of the cost vectors ${\bf c}_1,\dots,{\bf c}_{n - 1}$ and the realized actions $a_1,\dots,a_{n-1}$ of the first $n - 1$ time steps. An {\it adversary} for such an algorithm $\mathcal{A}$ specifies for each $n$ the cost vector ${\bf c}_n$, as a function of the probability distributions $x_1, \dots , x_n$ used by $\mathcal{A}$ on the first $n$ days and the realized actions $a_1, \dots , a_{n-1}$ of the first $n - 1$ days.
For example, elements of $A$ could represent different investment strategies, different driving routes between home and work, different wireless routers, etc. 

Rather than typically comparing the expected reward of an algorithm to that of the best action {\it sequence} in hindsight, we compare it to the reward incurred by the {best fixed action} in hindsight. Namely, we change our benchmark from $\sum_{n=1}^T \min_{a \in A} c_n(a)$ to $\min_{a \in A} \sum_{n=1}^T c_n(a)$.
\\

\noindent\textbf{Regret}: Fix  cost vectors ${\bf c}_1,\dots,{\bf c}_T$. The (expected) {\it regret} of  the (randomized) algorithm  $\mathcal{A}$  choosing actions according $x_1,\dots,x_T$ is
\begin{equation}
 \underbrace{\sum_{n=1}^T \E_{a_n \sim x_n} c_n(a_n)}_{\text{our algorithm}}- \underbrace{\min_{a \in A} \sum_{n=1}^T c_n(a)}_{\text{best fixed action}},
\end{equation}
where $\E_{a_n \sim x_n} c_n(a_n)$ expresses the expected cost of the algorithm in time period $n$, when an action $a_n\in A$ is chosen according to the probability distribution $x_n$.

In the game theory context, the cost vector ${\bf c}_n$ of the each player is incurred by playing the (congestion) game with the other players. Formally, the cost vector for each agent at time $n$ is $\vec{c}_n =  \big(\alpha N x_n, \beta N (1-x_n)\big)$. 
The expected accumulated cost  (of any of the symmetric infinitesimally small) agents in time periods $1,\dots, T$ is equal to $\sum_{n=1}^T \big(\alpha  N x^2_n+  \beta N (1-x_n)^2\big)$. The expected regret of the algorithm is given by the formula 
 
 $$\sum_{n=1}^T \big(\alpha N x^2_n+  \beta N (1-x_n)^2\big) -  \min_{a \in A} \sum_{n=1}^T c_n(a) , $$
 
 \noindent
 or  equivalently,
 
 $$\sum_{n=1}^T \big(\alpha N x^2_n+  \beta N (1-x_n)^2\big) - \min \left\{\sum_{n=1}^T \alpha N x_n, \sum_{n=1}^T \beta N (1-x_n)\right\}. $$

\medskip



 \noindent{\textbf{Price of Anarchy:}} The \textit{Price of Anarchy} of a game is the ratio of the supremum of the social cost over all Nash equilibria divided by the social cost of the optimal state, where the social cost of a state is the sum of the costs of all agents. In our case, where the total flow (demand, or population size)  is $N$, and a fraction $x$ of the population adopts the first strategy, the social cost is $SC(x)= \alpha N^2 x^2 + \beta N^2 (1-x)^2$.

In non-atomic congestion games, it is well known that all equilibria have the same social cost. Moreover, for linear cost functions $c_1(x) = \alpha x, \ c_2(x)=\beta x$ it is straightforward to see that the Price of Anarchy is equal to one, as the unique equilibrium flow, $\beta$, also is the unique minimizer of the social cost, which attains a value\footnote{This statement is true under the normalization assumption $\alpha+\beta=1$.} of $N^2 \alpha \beta$. 

With above terminology in mind, we will study the time-average social cost given some initial condition. Since the time average of social cost may not converge to a specific value, research in algorithmic game theory typically focuses on
 the supremum over all initial conditions over all convergent subsequences.  If we only consider time sequences with vanishing time-average regret, their time averages normalized by the social cost of the optimal state is called the \textit{Price of Total Anarchy}, and, for non-atomic congestion games, it is equal to the Price of Anarchy, which is equal to $1$ in our case, suggesting optimal system performance \cite{blum2006routing}. Since MWU is not run with a decreasing step-size, the time-average regret may not vanish and a more careful analysis is needed. Lastly, taking a dynamical systems point of view, we will also study typical dynamical trajectories, since simulations suffice to identify the limits of time-averages of these trajectories, which occur for initial conditions with  positive Lebesgue measure.  

As we will show, as the total demand increases, the system will bifurcate away from the Nash equilibrium; and the time-average social cost will be strictly greater than its optimal value (in fact it can be artibtrarily close to its worst possible value). The Price of Anarchy result is no longer be predictive of the true system's non-equilibrium behavior. We now state the formula that will be used later. Under the assumption $\alpha+\beta=1$, we define the normalized time-average social cost as follows: 


\begin{equation}
\label{eqn: PoA_var}	 
\frac{\text{Time-average social cost}}{\text{Optimum social cost}}=
\frac{ \frac1T \sum_{n=1}^T \big(\alpha N^2 x^2_n+  \beta N^2 (1-x_n)^2\big)}{ N^2 \alpha \beta} =  \frac{ \frac1T\sum_{n=1}^T \big( x_n^2 -2\beta x_n +\beta \big)}{\beta (1-\beta)}.  
\end{equation}


\section{Limit cycles and chaos, with time-average convergence to Nash equilibrium}\label{s:1d}

This section discusses the behavior of the one-dimensional map defined by ~\eqref{map}, and its remarkable time-average properties, which we will later employ to analyze the time-average regret and the normalized time-average social cost in Sections \ref{s: regret} and \ref{s: socialcost}.
The map generated by non-atomic congestion games here reduces to the map studied in \cite{CFMP}, in which two-agent linear congestion games are studied. Up to redefinition of the parameters as well as with the symmetric initial conditions, i.e., on the diagonal, the one-dimensional map in the two scenarios are identical. Thus in this section, we restate some key properties of the map.  For the proofs we refer the reader to \cite{CFMP}.

We will start by investigating the dynamics under the map
~\eqref{map} given by 
\begin{equation}\label{1dmap}
f_{a,b}(x)=\frac{x}
{x+(1-x)\exp(a(x-b))},
\end{equation}
with $a>0$, $b\in (0,1)$. This map is generally asymmetric, unless $b=1/2$, see the middle column of Figure~\ref{fig: cobwebb_b0p7}.
It has three fixed points: $0$, $b$ and $1$.
The derivatives at the three fixed points are
\[
f'_{a,b}(0)=\exp(ab),\ \ f'_{a,b}(1)=\exp(a(1-b)),\ \ f'_{a,b}(b)=ab^2-ab+1.
\]
Hence, the fixed points 0 and 1 are repelling, while $b$ is repelling whenever $a>\frac2{b(1-b)}$.
\begin{figure}[t]
  \centering
   {\includegraphics[width = 0.8\textwidth]{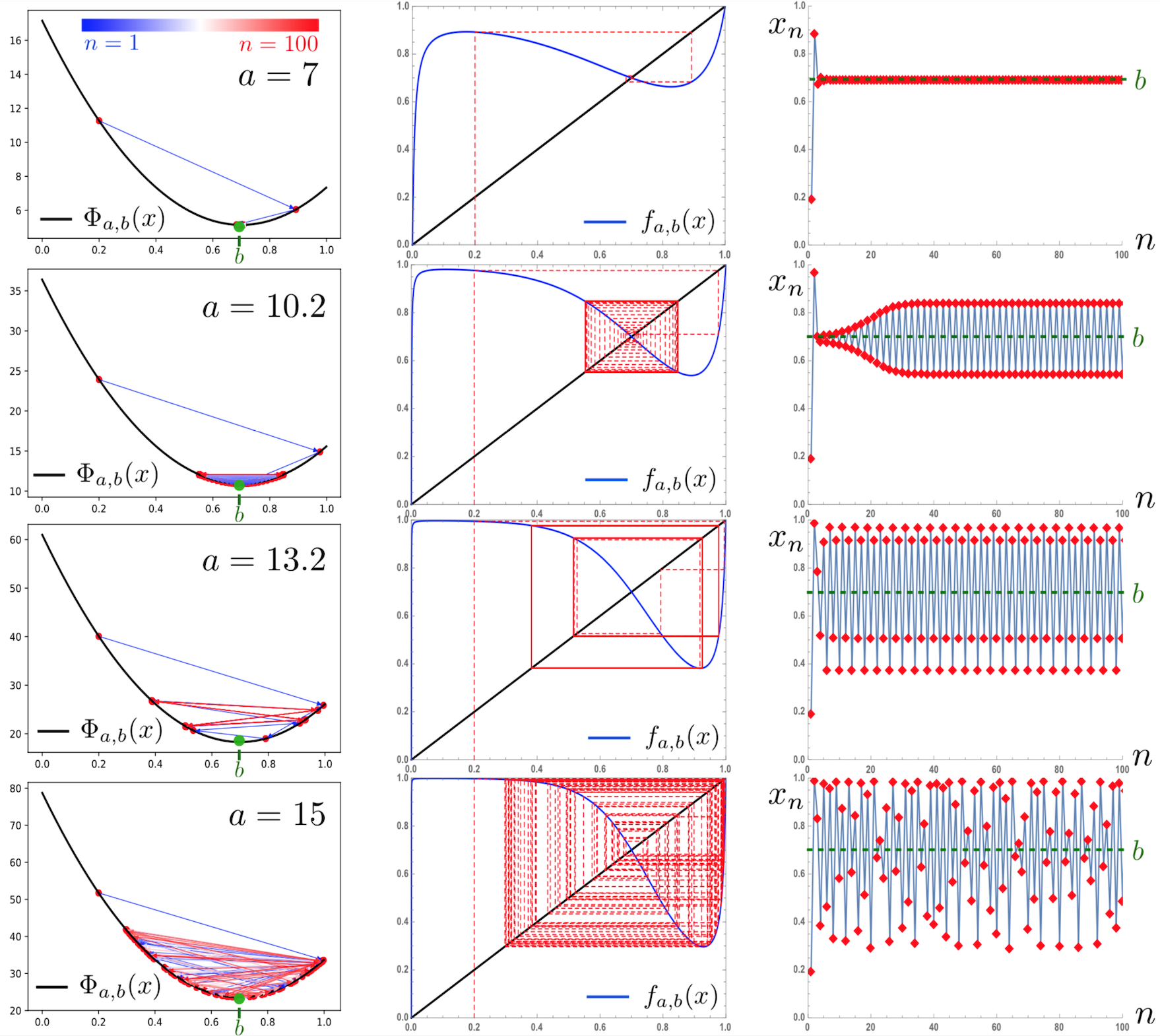}}
   \caption{ Population increase drives period-doubling instability and chaos. Although our congestion game has an associated convex potential function $\Phi_{a,b}(x) = \frac{N^2}{2}\left(\alpha x^2 + \beta (1-x)^2 \right) = \frac{a^2}{2}\left( (1-b)x^2 + b(1-x)^2\right)$ whose unique global minimum is the Nash equilibrium $b$ (without loss of generality, we set $\alpha+\beta = 1$ and $\epsilon = 1-1/e$ so that $a = N$ and $b = \beta$), MWU at large $a$, or equivalently large $N$ with a fixed $\epsilon$, do not converge to the equilibrium, unlike a gradient-like update with a small step size. A line with an arrow connecting $\Phi_{a,b}(x_n)$ to $\Phi_{a,b}(x_{n+1}) = \Phi_{a,b}(f_{a,b}(x_n))$ in the left column is encoded with the color representing the timestep $n$. Later times are shown in red, while earlier times are shown in blue.  Cobweb diagrams of the map $f_{a,b}$  are shown in the middle column, while the dynamics of the map are shown on the right column.  From top to bottom, values of $a$ increase while $b$ remains fixed at $0.7$, demonstrating population size-driven instability. At these parameter values, the map $f_{a,b}$ is bimodal (blue curves in the middle column). For small $a$ (top row), the dynamics converge to the Nash equilibrium $b$. As $a$ increases, the dynamics converge to a period $2$ attractor (second row), a period $4$ attractor (third row), and a chaotic attractor (bottom row). As shown in Section 3, however, the time average of these orbits is exactly the Nash equilibrium $b$, represented by the horizontal green dashed lines on the right column. The initial condition here is set to $x_0 = 0.2.$ The bifurcation diagram associated with $b=0.7$ is shown in Fig. \ref{fig: regbound_b0p7}.}
  \label{fig: cobwebb_b0p7}
\end{figure}

The critical points of $f_{a,b}$ are solutions to $ax^2-ax+1=0$. Thus,
if $0<a\le 4$, then $f_{a,b}$ is strictly increasing. If $a>4$, it has
two critical points
\begin{equation}\label{xlxf}
x_l=\frac 12\left(1-\sqrt{1-\frac 4a}\right), \ \ \ 
x_r=1-x_l=\frac 12\left(1+\sqrt{1-\frac 4a}\right)
\end{equation}
so $f_{a,b}$ is bimodal.

Let us investigate regularity of $f_{a,b}$. It is clear that it is
analytic. However, nice properties of interval maps are guaranteed not
by analyticity, but by the negative Schwarzian derivative. Let us recall
that the Schwarzian derivative of $f$ is given by the formula
\[
Sf=\frac{f'''}{f'}-\frac32\left(\frac{f''}{f'}\right)^2.
\]
A ``metatheorem'' states that almost all natural noninvertible
interval maps have negative Schwarzian derivative. Note that if $a\le
4$ then $f_{a,b}$ is a homeomorphism, so we should not expect negative
Schwarzian derivative for that case.

\begin{proposition}\label{nS}
If $a>4$ then the map $f_{a,b}$ has negative Schwarzian derivative.
\end{proposition}

For maps with negative Schwarzian derivative each attracting or
neutral periodic orbit has a critical point in its immediate basin of
attraction. Thus, we know that if $a>4$ then $f_{a,b}$ can have at
most two attracting or neutral periodic orbits.

\subsection{Time-average convergence to Nash equilibrium $b$}

While we know that the fixed point $b$ is often repelling, especially
for large values of $a$, we can show that it is attracting in a time-average sense.

\begin{definition}\label{d:Cesaro}
For an interval map $f$ a point $p$ is \emph{Ces\`{a}ro attracting} if there is a set $U$ such that for every $x\in U$ the averages
\[
\frac1T\sum_{n=0}^{T-1}f^n(x)
\]
converge to $p$.
\end{definition}

We can show that $b$ is globally Ces\`{a}ro attracting. Here by
``globally'' we mean that the set $U$ from the definition is the
interval $(0,1)$.

\begin{theorem}\label{t:Cesaro}
For every $a>0$, $b\in(0,1)$ and $x\in(0,1)$ we have
\begin{equation}\label{e:Cesaro}
\lim_{T\to\infty}\frac1T\sum_{n=0}^{T-1}f_{a,b}^n(x)=b.
\end{equation}
\end{theorem}

\begin{corollary}\label{cmper}
For every periodic orbit $\{x_0,x_1,\dots,x_{T-1}\}$
of $f_{a,b}$ in $(0,1)$ its center of mass (time average)
\[
\frac{x_0+x_1+\dots+x_{T-1}}T
\]
is equal to $b$.
\end{corollary}

Applying the Birkhoff Ergodic Theorem (see Section \ref{s:entropy}), we get a stronger corollary.

\begin{corollary}\label{cmmeas}
For every probability measure $\mu$, invariant for $f_{a,b}$ and such
that $\mu(\{0,1\})=0$, we have
\[
\int_{[0,1]} x\; d\mu=b.
\]
\end{corollary}

These statements show that the time average of the orbits generated by $f_{a,b}$ converges exactly to the Nash equilibrium $b$. Next, we discuss what happens as we fix $b$ and increase the total demand by letting $a$ grow large. 


\subsection{Periodic orbits and chaotic behavior}
When $b=0.5$, the coefficients of the cost functions are identical, i.e., 
$\alpha=\beta$. 

\begin{theorem}\label{trajf}
If  $0<a\le 8$ then $f_{a,0.5}$-trajectories of all points of $(0,1)$
converge to the fixed point $0.5$. If $a>8$ then $f_{a,0.5}$ has a periodic
attracting orbit $\{\sigma_a,1-\sigma_a\}$, where $0<\sigma_a<0.5$. This
orbit attracts trajectories of all points of $(0,1)$, except countably
many points, whose trajectories eventually fall into the repelling
fixed point $0.5$.
\end{theorem}

Now we proceed with the case when $b\neq 0.5$, that is, when the
cost functions differ. As we noticed previously in Section~\ref{s:1d}, if $a\le
4$ then $f_{a,b}$ is strictly increasing and has three fixed points: 0
and 1 repelling and $b$ attracting. Therefore in this case
trajectories of all points of $(0,1)$ converge in a monotone way to
$b$.
We know that the fixed point $b$ is repelling if and only if
$a>\frac2{b(1-b)}$.
This is a simple situation, so we turn to the case of large $a$. We
fix $b\in(0,1)\setminus\{0.5\}$ and let $a$ go to infinity. We will
show that if $a$ becomes sufficiently large (but how large, depends on
$b$), then $f_{a,b}$ is Li-Yorke chaotic and has periodic orbits of
all possible periods.

\begin{definition}[Li-Yorke chaos]  \label{LYchaos-def}
Let $(X,f)$ be a dynamical system and $(x,y)\in X\times X$.
We say that $(x,y)$ is a \emph{Li-Yorke pair} 
if
\begin{align*}
\liminf_{n\to\infty} dist (f^n(x),f^n(y))&=0,\\
\limsup_{n\to\infty} dist (f^n(x),f^n(y))&>0.
\end{align*}
A dynamical system $(X,f)$ is \emph{Li-Yorke chaotic} if there is an uncountable set $S\subset X$ (called \emph{scrambled set}) such that every pair $(x,y)$ with $x,y\in S$ and $x\neq y$ is a Li-Yorke pair.\footnote{Intuition behind this definition as well as other properties of chaotic behavior of dynamical systems are discussed in Section \ref{s:entropy}.} 
\end{definition}

The crucial ingredient of this analysis is the existence of periodic orbit of period 3.

\begin{theorem}\label{per3}
If $b\in (0,1)\setminus\{0.5\}$, then there exists $a_b$ such that if
$a>a_b$ then $f_{a,b}$ has periodic orbit of period 3.
\end{theorem}

By the Sharkovsky Theorem (\cite{sha}, see also~\cite{liyorke}), existence
of a periodic orbit of period 3 implies existence of periodic orbits
of all periods, and by the result of~\cite{liyorke}, it implies that the
map is Li-Yorke chaotic.
Thus, we get the following corollary:

\begin{corollary}\label{chaos}
If $b\in (0,1)\setminus\{0.5\}$, then there exists $a_b$ such that if
$a>a_b$ then $f_{a,b}$ has periodic orbits of all periods and is
Li-Yorke chaotic.
\end{corollary}

This result has a remarkable implication in non-atomic routing games. Recall that the parameter $a$ expresses the normalized total flow/demand; thus, Corollary~\ref{chaos} implies that when the game is asymmetric,
i.e. when an interior equilibrium flow is not the $50\%-50\%$ split, increasing the total demand of the system will inevitably lead to chaotic behavior, regardless of the form of the cost functions.\footnote{In fact we can strengthen this result, see Section \ref{s:entropy}.} In other words, the emergence of chaos at sufficiently large demands is a robust phenomenon.

\section{Analysis of time-average regret}\label{s: regret}
 In the previous section, we discussed the time average convergence to Nash equilibrium for the map $f_{a,b}$. We now employ this property to investigate the time-average regret from learning with MWU. 
 
 \begin{theorem} \label{thm: reg_var}
The limit of the time-average regret is the total demand $N$ times the
limit of the observable $(x-b)^2$ (provided this limit exists).
That is 
\begin{equation}
\label{eqn: reg_var}
\lim_{T \rightarrow \infty}\frac {R_T}{T} = N \left(\lim_{T
  \rightarrow \infty}\frac 1T \sum_{n=1}^T (x_n-b)^2\right).
\end{equation}
\end{theorem}

\begin{proof}
Recall that the time-average regret is
\begin{equation}\label{e51}
\frac 1T R_T=\frac 1T\sum_{n=1}^T \big(\alpha Nx^2_n+  \beta
N(1-x_n)^2\big) - \min\left \{\frac 1T\sum_{n=1}^T \alpha N x_n, \frac
1T\sum_{n=1}^T \beta N (1-x_n) \right\}
\end{equation}
Consider
\begin{equation}\label{eqn: var_pf_step}
\frac 1T \left(\sum_{n=1}^T\alpha x_n-\sum_{n=1}^T\beta (1-x_n)\right)
=\frac 1T \sum_{n=1}^T\left[(\alpha+\beta)x_n-\beta\right]=
(\alpha+\beta)\left(\frac{1}{T} \sum_{n=1}^Tx_n
-\frac{\beta}{\alpha+\beta}\right).
\end{equation}
The quantity $\frac{\beta}{\alpha+\beta}$ is the system equilibrium
$b$. Without loss of generality we assume as mentioned earlier
$\alpha+\beta=1$. Then, $b=\beta$ and $\lim_{T \rightarrow \infty}
\left(\frac{1}{T}\sum_{n=1}^T x_n \right) = b$, by
Theorem~\ref{t:Cesaro}.

Therefore, in the limit $T \rightarrow \infty$, two terms in
$\min$-term of~\eqref{e51} coincide and we have by substituting
$\alpha= 1-\beta$ and remembering that $\beta=b$
\begin{eqnarray}
\lim_{T \rightarrow \infty}\frac {R_T}{T}&=& \lim_{T \rightarrow
  \infty}\frac NT \sum_{n=1}^T \left((1-\beta) x^2_n+ \beta (1-x_n)^2-
\beta (1-\beta) \right) \nonumber\\ &=& \lim_{T \rightarrow
  \infty}\frac NT \sum_{n=1}^T\left( x_n^2 - 2\beta x_n +\beta^2
\right)= \lim_{T \rightarrow \infty}\frac NT \sum_{n=1}^T(x_n-b)^2.
\nonumber
\end{eqnarray}
\end{proof}

Observe that if $x$ is a generic point of an ergodic invariant
probability measure $\mu$, then the time limit of the observable
$(x-b)^2$ is equal to its space average $\int_0^1(x-b)^2\;d\mu(x)$.
This quantity is the variance of the random variable identity (we will
denote this variable $X$, so $X(x)=x$) with respect to $\mu$. Typical
cases of such a measure $\mu$, for which the set of generic points has
positive Lebesgue measure, are when there exists an attracting
periodic orbit $P$ and $\mu$ is the measure equidistributed on $P$,
and when $\mu$ is an ergodic invariant probability measure absolutely
continuous with respect to the Lebesgue measure \cite{denker2006ergodic}. In analogue to the
family of quadratic interval maps \cite{lyubich2000quadratic}, we have reasons to expect that for
the Lebesgue almost every pair of parameters $(a,b)$ Lebesgue almost
every point $x\in(0,1)$ is generic for a measure of one of those two
types.

\vskip 0.2in

\noindent{\bf Upper bound for time-average regret}: 
Let $a>4$ and $b\in (0,1)$. Recall from \eqref{xlxf} that $f_{a,b}$ has two critical points 
$x_l=\frac 12\left(1-\sqrt{1-\frac 4a}\right)$ and $x_r=1-x_l=\frac 12\left(1+\sqrt{1-\frac 4a}\right)$.

Let \[y_{\min}=f_{a,b}(x_r),\;\;\;y_{\max}=f_{a,b}(x_l).\]

\begin{lemma} \label{interval}
For $a>\frac{1}{b(1-b)}$   
the interval $I=[y_{\min},y_{\max}]$ is invariant and globally absorbing on $(0,1)$.
\end{lemma}
\begin{proof}
Fix $b\in (0,1)$. Simple calculations show that $b\in (x_l,x_r)$ if and only if $a>1/b(1-b)$.\footnote{Therefore  $b\in (x_l,x_r)$ when $x=b$ is repelling (for $a>2/b(1-b)$).} 
From now  we assume that $a>1/b(1-b)$. Recall that $b$ is a fixed point of $f_{a,b}$ so we have $b=f_{a,b}(b)\in f_{a,b}([x_l,x_r])=[y_{\min},y_{\max}]=I$. Therefore $b\in I\cap (x_l,x_r)$ and $I\cap (x_l,x_r)\neq \emptyset$. 
Because $f_{a,b}$ is decreasing between the critical points, we have  $f_{a,b}'(b)=ab^2-ab+1<0$.  The latter and the uniqueness of a fixed point in $(0,1)$ implies that  $f_{a,b}(x)>x$ for $x\in (0,b)$, and $f_{a,b}(x)<x$ for $x\in (b,1)$.

Obviously, if $x\in I\cap (x_l,x_r)$, then $f_{a,b}(x)\in f_{a,b}([x_l,x_r])=I$. For $x\in [y_{\min},x_l)$ we have $f_{a,b}(x)<f_{a,b}(x_l)=y_{\max}$. Suppose that $f_{a,b}(x)<y_{\min}$, then $f_{a,b}(x)<y_{\min}\leq x$ but it is impossible because $f_{a,b}(x)>x$ for $x\in (0,b)$. Thus $f_{a,b}([y_{\min},x_l])\subset I$. The same reasoning shows that  $f_{a,b}([x_r,y_{\max}])\subset I$. Thus $f_{a,b}(I)\subset I$.


Now let $x\in (0,x_l)$. Obviously $f_{a,b}(x)<y_{\max}$ and because $x_l<b$ 
we have $f_{a,b}(x)>x$ for $x<x_l$. To show that the orbit of $x$ falls eventually into $I$, it is sufficient to show that there exists $n$ such that $f_{a,b}^n(x)>y_{\min}$.
Suppose that there is no such $n$, that is $f_{a,b}^n(x)<y_{\min}$ for all $n$. The sequence $(f_{a,b}^n(x))_{n\geq 1}$ is increasing and bounded, so it has a limit. Denote the limit by $c \in (0,x_l)$. Then $f_{a,b}(c)$ has to be equal $c$, but this contradicts the fact that on $(0,x_l)$ there is no fixed point of $f_{a,b}$. Thus, orbits of every point from $(0,x_l)$ will eventually fall into the invariant set $I$. Similar reasoning will show that orbits of every point from $x\in (x_r,1)$ will eventually fall into $I$.

\end{proof}



Lemma \ref{interval}  implies the following bound for the variance.
\begin{remark}
For  $a>\frac{1}{b(1-b)}$    the variance of $X$ is bounded above by 
\begin{equation*}
\text{Var}(X) = \lim_{T \rightarrow \infty}\left(\frac 1T \sum_{n=1}^T (x_n-b)^2\right)\leq (y_{\max}-b)(b-y_{\min}).	
\end{equation*}

\end{remark}

To see how time-average regret changes as one increases the total demand $N$, we let as usual $\epsilon = 1 - 1/e$, so that $a = N \ln \left( \frac{1}{1-\epsilon} \right) = N$. 

\begin{theorem} For  $N>\frac{1}{b(1-b)}$   the time-average regret is bounded above by
\begin{equation}
\label{eqn: regbound}
	\lim_{T \rightarrow \infty}\frac {R_T}{T} = N \left(\lim_{T \rightarrow \infty}\frac 1T \sum_{n=1}^T (x_n-b)^2\right)\le N (y_{\max} -b)(b - y_{\min}).
\end{equation}	
\end{theorem}

\begin{figure}[t]
  \centering
   {\includegraphics[width = 0.9\textwidth]{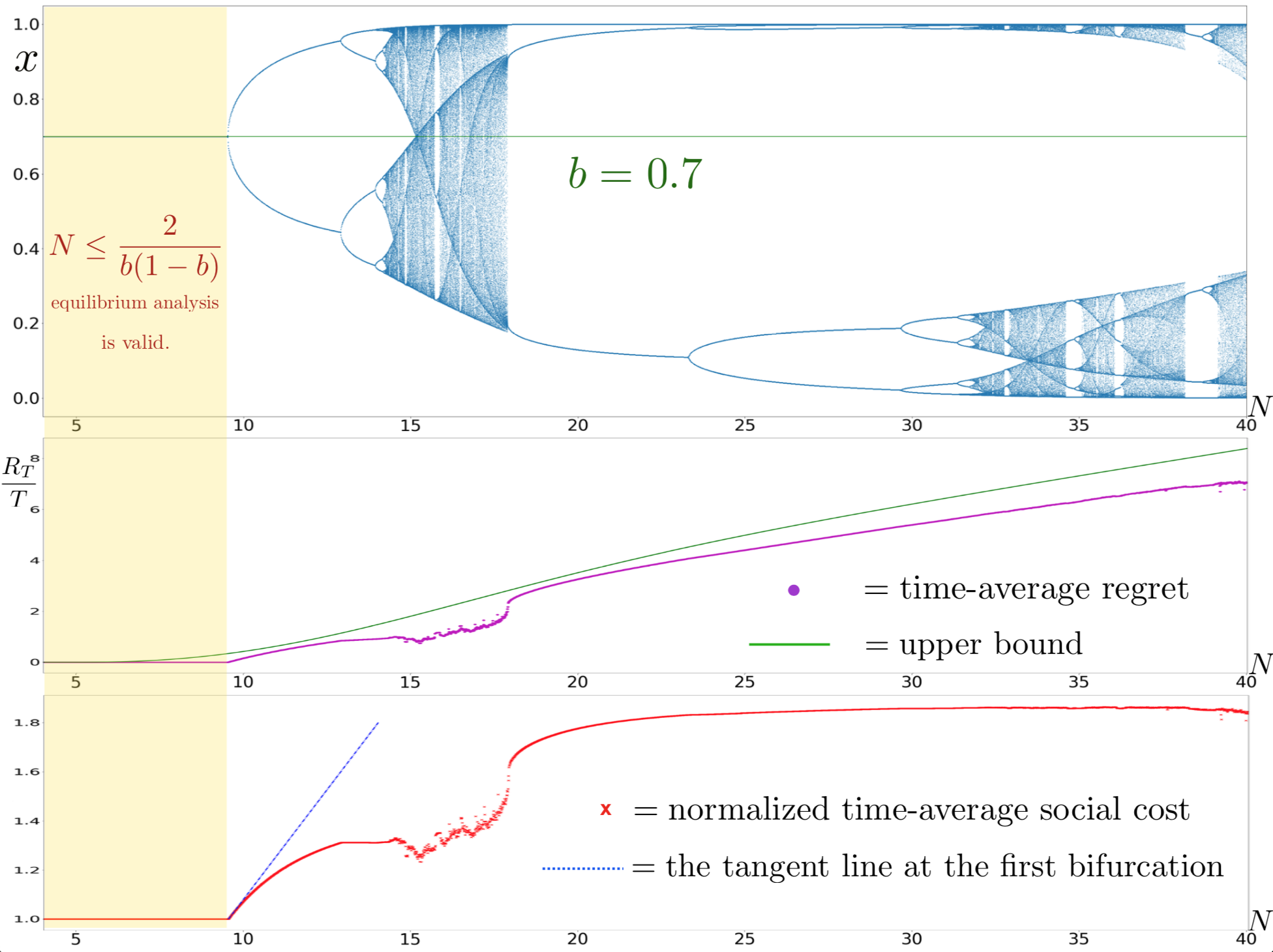}}
   \caption{ Bifurcation diagram (Top) demonstrates instability of the routing game driven by the increase in total demand $N.$  The Nash equilibrium here is set to $b = 0.7$. As usual, we fix the learning rate $\epsilon = 1 - 1/e$ so that $a = N$ for simplicity. At small $N$, the dynamics converges toward the fixed point $b$, which is the Nash equilibrium. However, as $N$ exceeds the carrying capacity of $N^*_b = 2/b(1-b)$, the Nash equilibrium becomes repelling and the dynamics no longer converge to it. The period-doubling route to chaos begins.  Remarkably, the time-average of all orbits is exactly $b$, as is evident from the green line that tracks the center of masses of the blue orbits. (Middle) The time-average regret $\frac 1T R_T$ is shown in purple. It {\it suddenly} becomes strictly positive at the first bifurcation, consistent with the prediction of (\ref{eqn: reg_var}) which states that the time-average regret is proportional to the fluctuations from the Nash equilibrium. 
  The green line shows our upper bound on the time-average regret from Equation (\ref{eqn: regbound}). (Bottom) The normalized time-average social cost (i.e., time average social cost divided by optimum) also {\it suddenly} becomes greater than $1$ at the first bifurcation, consistent with the prediction of Equation (\ref{eqn: PoA_var}). Hence, the Price of Anarchy bound (of $1$) is only a valid upper bound for the system inefficiency before the first bifurcation arises. 
  The time average cost of the non-equilibrating dynamics is proportional to its fluctuations away from the Nash equilibrium. The rate at which the normalized time-average social cost increases above unity at the first bifurcation is depicted by the tangent line (dashed blue), which is calculated from Equation (\ref{eqn: derivative_POA}).}
  \label{fig: regbound_b0p7}
\end{figure}


\section{Analysis of time-average social cost}\label{s: socialcost}

We begin this section with an extreme scenario of the time-average social cost; for a symmetric equilibrium flow ($b=0.5$), the time-average social cost can be arbitrarily close to its worst possible value! 
In contrast to the optimal social cost attained at the equilibrium $b = 0.5$, the long-time dynamics alternate between the two periodic points of the limit cycle of period two, which, at large population size, can approach $1$ or $0$ arbitrarily closely, see Figure \ref{fig: intro_summary} (top) or Figure \ref{fig: regbound_b0p5}. This means almost all users will occupy the same route, while simultaneously alternating between the two routes. The time-average social cost thus becomes worse at larger population sizes, approaching its worst possible value in the limit of infinite population.


\begin{theorem} 
\label{t:SC}
For $b=0.5$, the  time-average social cost can be arbitrarily close to its worst possible value for a sufficiently large $a$, i.e. for a sufficiently large population size\footnote{Recall from (\ref{var2}) that $a = N\ln\left( \frac{1}{1-\epsilon} \right)$.}. Formally, for any $\delta>0$, there exists an $a$ such that, for any initial condition $x_0$, except countably
many points, whose trajectories eventually fall into the fixed point $b$ we have 
$$\liminf_{T\to \infty} \frac{1}{T}\sum_{n=1}^T SC(x_n) > \max_x SC(x) -\delta$$ 
\end{theorem}

\begin{proof}
For a symmetric equilibrium $b=0.5$, the two cost functions increase with the loads at the same rate $\alpha=\beta$. 
Recall from Section \ref{sec: regretPoASC} that the social cost when  fraction $x$ of the population adopts the first strategy is $SC(x)= \alpha N^2 x^2 + \alpha N^2 (1-x)^2$. This strictly convex function attains its minimum at the equilibrium $x=b=0.5$, and its maximum of $\alpha N^2$ at $x=0$ or $x=1$. By Theorem \ref{trajf} we know that for $a>8$ 
there exists a periodic attracting orbit $\{\sigma_a,1-\sigma_a\}$, where $0<\sigma_a<0.5$. This
orbit attracts trajectories of all points of $(0,1)$, except countably many points, whose trajectories eventually fall into the repelling fixed point $0.5$.  
To establish that for all trajectories attracted by the orbit  $\{\sigma_a,1-\sigma_a\}$ the time-average limit of the social cost can become arbitrarily close to $\alpha N^2$, it suffices to show that the distance of the two periodic points (of the unique attracting period-$2$ limit cycle) to the nearest boundary goes to zero as $a \rightarrow \infty$.  
Thus, it suffices to show that given any $\delta>0$, there exist an $a$ such that $\sigma_a<\delta$.

For brevity, we denote the map $f_{a,0.5}$ by $f_a$. Then, since $\sigma_a$ is a periodic point of a limit cycle of period 2, we have $f^2_a(\sigma_a) = \sigma_a$. The last equality implies $f_{a}(\sigma_a)=1-\sigma_a$, which, after simple calculations, implies 
\[ \left(\frac{\sigma_a}{1- \sigma_a}\right)^2 = \exp[a(\sigma_a-0.5)].\]

Consider the function $\phi(x)= \left(\frac{x}{1- x}\right)^2 - \exp[a(x-0.5)]$, then $\phi(0)= - \exp[-0.5a]<0$.
On the other hand, for any $\delta \in (0,0.5)$, $\phi(\delta)= \big(\frac{\delta}{1- \delta}\big)^2 - \exp[a(\delta-0.5)]> \frac12  \big(\frac{\delta}{1- \delta}\big)^2>0 $ for a sufficiently large $a >0.$ The intermediate value theorem implies $\sigma_a \in (0, \delta),$  and the theorem follows.
\end{proof}

More generally, when the equilibrium flow is asymmetric ($b \neq 0.5$), we can relate the normalized time-average social cost to the non-equilibrium fluctuations from the equilibrium flow. From (\ref{eqn: PoA_var}), we obtain
\begin{equation}
\label{eqn: time-avg-sc-var}	
\text{normalized time-average social cost} =  \frac{ \frac1T\sum_{n=1}^T \big( x_n^2 -2\beta x_n +\beta \big)}{\beta (1-\beta)}  = 1 + \frac{\text{Var}(X)}{\beta(1-\beta)}. 
\end{equation}
If the dynamics converges to the fixed point $b$, the variance vanishes and the normalized time-average social cost coincides with the Price of Anarchy which is 1. However, as the total demand $N$ increases, the system suddenly bifurcates at $N = N^*_b \equiv 2/b(1-b)$, which is the carrying capacity of the network. Above the carrying capacity $N^*_b$, the system is non-equilibrating, and the variance becomes positive. As a result, the normalized time-average social cost becomes greater than 1. The Price of Anarchy prediction error, i.e., (normalized time-average social cost)-(Price of Anarchy)$>0$,
 depends on how fast the variance {\it suddenly} increases at the first bifurcation point, which we analyze next.

\subsection{Analysis of variance spreading at the first period-doubling bifurcation}

We study the behavior of the variance as $a$ crosses the period doubling bifurcation point. We first consider the model situation, where the map is
$g(x)=(\gamma_1-1)x+\gamma_2 x^2+\gamma_3 x^3$. Note that $\gamma_1$ and $\gamma_2$ here are not the coefficients of the cost functions. In this model situation, the bifurcation occurs at $\gamma_1=0$. If $\gamma_1>0$ then the fixed point $x=0$ is attracting, and as $\gamma_1<0$ then it is repelling, but under some conditions on the coefficients there is an attracting periodic orbit of period 2.
We are interested only at the limit behavior as $\gamma_1$ goes to zero, in a small neighborhood of $x=0$. Therefore we may ignore all powers of $x$ larger than 3 and all powers of $\gamma_1$ larger than 1.
Period 2 points are non-zero solutions of the equation
\[
x=(\gamma_1-1)[(\gamma_1-1)x+\gamma_2 x^2+\gamma_3 x^3]+\gamma_2[(\gamma_1-1)x+\gamma_2 x^2+\gamma_3 x^3]^2+\gamma_3[(\gamma_1-1)x+\gamma_2 x^2+\gamma_3 x^3]^3.
\]
Ignoring higher order terms and dividing by $x$, we get the equation
\[
(2\gamma_2^2\gamma_1-2\gamma_2^2+4\gamma_1\gamma_3-2\gamma_3)x^2-\gamma_1\gamma_2 x-2\gamma_1=0.
\]
Its discriminant is (after ignoring higher order terms in $\gamma_1$)
\[
\Delta=-16\gamma_1(\gamma_2^2+\gamma_3),
\]
so
\[
x=\frac{\gamma_1\gamma_2\pm 4\sqrt{-\gamma_1(\gamma_2^2+\gamma_3)}}{4(\gamma_2^2\gamma_1-\gamma_2^2+2\gamma_1\gamma_3-\gamma_3)}.
\]
Assume now that $\gamma_2^2+\gamma_3>0$. This is equivalent to $Sg(0)<0$, so we will be able to apply it to our map (see Proposition \ref{nS}). If $\gamma_1$ is close to zero, in the numerator $\gamma_1$ is negligible compared to $\sqrt{\gamma_1}$, and in the denominator $\gamma_1$ is negligible compared to a constant. Thus, approximately we have
\[
x=\pm\sqrt{\frac{-\gamma_1}{\gamma_2^2+\gamma_3}}.
\]
Therefore, the variance is $\text{Var}(X) = \frac{-\gamma_1}{\gamma_2^2+\gamma_3}$.

After Taylor expanding the map of Equation (\ref{map}) around the fixed point $b$ and comparing the cofficients to those of $g$, we obtain $\gamma_1 = 2 + a b(b-1)$, $\gamma_2 = a(b-\frac 12)(1+ab(b-1))$ and $\gamma_3 = a(1+a(\frac 16 + b(b-1))(3+ab(b-1))).$
Recalling the first bifurcation occurs at $a^*_b = 2/b(1-b)$, we thus deduce the (right) derivative of the variance with respect to $a$ at the first period-doubling bifurcation: 
\begin{equation}
\frac{d \text{Var}(X)}{da} \Big |_{a = a^{*+}_b} = -\frac{d\left(\frac{\gamma_1}{\gamma_2^2+\gamma_3}\right) }{d a}  \Big |_{a = a^{*+}_b}= \frac{3 b^3(1 - b)^3 }{2 - 6b (1 - b) },  
\end{equation}
which is a unimodal function in the interval $[0,1]$ that is symmetric around $b=0.5$, at which the maximum $0.09375$ is attained.

This allows us to deduce how fast the normalized time-average social cost increases at the first bifurcation, signaling how equilibrium Price of Anarchy metric fails as we increase $a$, or equivalently increase $N$. Namely, from Equation (\ref{eqn: PoA_var}), one finds that the derivative of the normalized time-average social cost with respect to $a$ reads 
\begin{equation}
\label{eqn: derivative_POA}
	\frac{d}{da}\left(\text{normalized time-average social cost}\right) \ \Big |_{a = a^{*+}_b} = \frac{1}{b(1-b)}\frac{d \text{Var}(X)}{da} \Big |_{a = a^{*+}_b} = \frac{3 b^2(1 - b)^2 }{2 - 6b (1 - b) }.
\end{equation}  
When $a < 2/b(1-b)$, the system equilibrates and the normalized time-average social cost is unity. However, when $a$ exceeds $ 2/b(1-b)$, the system is out of equilibrium, and normalized time-average social cost {\it suddenly} increases with a finite rate, given by Equation (\ref{eqn: derivative_POA}). At the first period-doubling bifurcation, the second-derivative with respect to $a$ becomes {\it discontinuous}, akin to the second order phase transition phenomena in statistical physics.
Fig. \ref{fig: regbound_b0p7} confirms the prediction of Equation (\ref{eqn: derivative_POA}).

Likewise, as the variance becomes positive, the time-average regret also becomes non-zero. At the first period-doubling bifurcation, the time-average regret given by Equation (\ref{eqn: reg_var}) {\it suddenly} increases with $a$ at the rate (where we use our typical normalization $a=N$) 
\begin{equation}
\label{eqn: derivative_regret}
	\frac{d}{da}\left(\text{time-average regret}\right) \ \Big |_{a = a^{*+}_b} = \frac{d \left( a \text{Var}(X)\right)}{da} \Big |_{a = a^{*+}_b} = \frac{3 b^2(1 - b)^2 }{1 - 3b (1 - b) }.
\end{equation} 
Therefore, at the first period-doubling bifurcation, where the equilibrium analysis begins to breakdown, the following equality holds
\begin{equation}
	\frac{d}{da}\left(\text{time-average regret}\right) \ \Big |_{a = a^{*+}_b} = 2\frac{d}{da}\left(\text{normalized time-average social cost}\right) \ \Big |_{a = a^{*+}_b}.
\end{equation}


\section{Chaos and Birkhoff Ergodic Theorem}
\label{s:entropy}
This section familiarizes the reader with key concepts from dynamical systems necessary for this work, e.g., chaotic behavior, absolutely continuous invariant measures, topological entropy, and ergodic theorem.

It seems that there is no universally accepted definition of chaotic behavior of a dynamical system. Most definitions of chaos concern one of the following aspects:
\begin{itemize}
\item complex behavior of trajectories, such as Li-Yorke chaos;
\item fast growth of the number of distinguishable orbits of length $n$, such as having positive topological entropy;
\item existence of absolutely continuous invariant measures;
\item sensitive dependence on initial conditions, such as Devaney or Auslander-Yorke chaos;
\item recurrence properties, such as transitivity or mixing. 
\end{itemize}

In this article, the first three are crucial. 
Also, in the presence of chaos, studying precise single orbit dynamics can be intractable; we study the average behavior of trajectories instead. Thus, it is important to know whether the average converges. This is when ergodic theorems come into play. 
\subsection{Li-Yorke chaos and topological entropy}

The origin of the definition of Li-Yorke chaos (see Definition \ref{LYchaos-def}) is in the seminal Li and Yorke's article \cite{liyorke}.
Intuitively orbits of two points from the scrambled set have to gather themselves arbitrarily close and  spring aside infinitely many times but (if $X$ is compact) it cannot happen simultaneously for each pair of points. 
 Why should a system with this property be chaotic? 
Obviously existence of the scrambled set implies that orbits of points behave in unpredictable, complex way.
More arguments come from the theory of interval transformations, in view of which it was introduced. For such maps the existence of one Li-Yorke pair implies the existence of an uncountable scrambled set \cite{KuchtaSmital} and it is not very far from implying all other properties that have been called chaotic in this context, see e.g. \cite{Ruette}. In general, Li-Yorke chaos has been proved to be a necessary condition for many other 'chaotic' properties to hold. A nice survey of properties of Li-Yorke chaotic systems can be found in \cite{BHS}.


A crucial feature of the chaotic behavior of a dynamical system is exponential growth of the number of distinguishable orbits. This happens if and only if the topological entropy of the system is positive. In fact positivity of topological entropy turned out to be an essential criterion of chaos \cite{GW93}.
This choice comes from the fact that the future of a deterministic (zero entropy) dynamical system can be predicted if its past is known (see \cite[Chapter 7]{Weiss}) and  positive entropy is related to randomness and chaos.

For every dynamical system over a compact phase space, we can define a number $h(f)\in[0,\infty]$ called the topological entropy of transformation $f$. This quantity was first introduced by Adler, Konheim and McAndrew \cite{AKM} as the topological counterpart of metric (and Shannon) entropy.

For a given positive integer $n$ we define the $n$-th Bowen-Dinaburg metric on $X$, $\rho_n^f$ as \[\rho_n^f(x,y)=\max_{0\leq i<n} dist (f^i(x),f^i(y)).\]

We say that the set $E$ is $(n,\varepsilon)$-separated if $\rho_n^f(x,y)>\varepsilon$ for any distinct $x,y\in E$ and by $s(n,\varepsilon,f)$ we denote the cardinality of the most numerous $(n,\varepsilon)$-separated set for $(X,f)$.

\begin{definition}
The  topological entropy of $f$ is defined as \[h(f)=\lim_{\varepsilon\searrow 0}\limsup_{n\to\infty}\frac 1n \log s(n,\varepsilon,f).\]
\end{definition} 

We begin with the intuitive explanation of the idea. 
Let us assume that we observe the dynamical system with the precision $\varepsilon>0$, that is, we can distinguish any two points, only if they are apart by at least $\varepsilon$. Then, we are able to see only as many points as is the cardinality of the biggest $(1,\varepsilon)$-separated set. Therefore, after $n$ iterations we will see at most $s(n,\varepsilon,f)$ different orbits. If transformation $f$ is mixing points, then $s(n,\varepsilon,f)$ will grow. Taking upper limit over $n$ will give us the exponential ratio of asymptotic growth of number of (distinguishable) orbits,  and going with $\varepsilon$ to zero will give us the quantity which can be treated as a measure of exponential speed, with which the number of orbits grow (with $n$).

Both positive topological entropy and Li-Yorke chaos are local properties; in fact, entropy depends only on a specific subset of the phase space and is concentrated on the set of so-called nonwandering points \cite{Bow70}.  The question whether positive topological entropy implies
Li-Yorke chaos 
remained open for some time, but eventually it was shown to be true; see \cite{BGKM}. 
On the other hand, there are Li-Yorke chaotic interval maps with zero topological entropy (as was shown independently by Sm\'{\i}tal \cite{Smital} and Xiong \cite{Xiong}). 
For deeper discussion of these matters we refer the reader to the excellent surveys
by Blanchard \cite{Blanchard}, Glasner and Ye \cite{GY}, Li and Ye \cite{LY14} and Ruette's book \cite{Ruette}.


{\bf Entropy of $f_{a,b}$ is positive}: After this discussion we can show how entropy behaves for $f_{a,b}$. For any interval map, we have the following:

\begin{theorem}[\cite{Mis}]\label{pertop}
For an interval map $f$, the following assertions are equivalent:
\begin{itemize}
\item[i)] $f$ has a periodic point whose period is not a power of 2,
\item[ii)] the topological entropy of $f$ is positive.
\end{itemize}
\end{theorem}

Thus, Theorem \ref{pertop} combined with Corollary \ref{chaos} strenghten the latter. 
\begin{corollary}\label{chaosentropy}
If $b\in (0,1)\setminus\{1/2\}$, then there exists $a_b$ such that if
$a>a_b$ then $f_{a,b}$ has periodic orbits of all periods, positive topological entropy and is
Li-Yorke chaotic.
\end{corollary}

{\bf Calculating entropy}: In general, computing the entropy is not an easy task. However, in the context of interval maps, topological entropy can be computed quite straightforwardly --- it is equal to the exponential growth rate of the minimal number of monotone subintervals for $f^n$.

\begin{theorem}[\cite{MS}]
Let $f$ be a piecewise monotone interval map and, for all $n\geq 1$, let $c_n$ be the minimal cardinality of a monotone partition for $f^n$. Then \[h(f)=\lim_{n\to \infty}\frac 1n \log c_n =\inf_{n\geq 1} \frac 1n \log c_n.\]
\end{theorem}
Moreover, for piecewise monotone interval maps, the entropy computed with any
partition into intervals, on which the map is monotone, is the {\it topological
entropy} \cite[Prop. 4.2.3]{alseda2000combinatorial}.
This gives us a way to understand what positive entropy of $f_{a,b}$ means from a game-theoretic perspective. For $a>\frac{1}{b(1-b)}$ the map $f_{a,b}$ is a bimodal map with two critical points $x_l,x_r$  (defined in (\ref{xlxf})) and a (unique in $(0,1)$) equilibrium $b\in (x_l,x_r)$. Because $x$ is the probability of choosing the first strategy, we can say that if $x<x_l$ or $x>x_r$, then one of the strategies is overused and if $x$ is close to $b$, $x\in[x_l,x_r]$, then the system is approximately at equilibrium.
Now, we can take a partition  $\{[0,x_l),[x_l,x_r], (x_r,1]\}$  into three intervals on which $f_{a,b}$ is monotone. For every $x\in [0,1]$ and for every $n\geq 1$ we encode three events for the $n$-th iteration of $x$:   $\mathbf{x}[n]=A$ if the system is approximately at equilibrium, that is if $f_{a,b}^n(x)\in [x_l,x_r]$; $\mathbf{x}[n]=B$ if the second strategy is overused, that is when $f_{a,b}^n(x)\in [0,x_l)$ and $\mathbf{x}[n]=C$ if the first strategy is overused,  $f_{a,b}^n(x)\in (x_r,1]$. This way for every $x\in [0,1]$ we get an infinite sequence  $\mathbf{x}$ on the alphabet $\{A,B,C\}$.
Now, the fact that $h(f_{a,b})>0$ implies that the number of different blocks of length $n$, which we can observe looking at different $ \mathbf{x}$ we generated this way, will grow exponentially.


\subsection{Invariant measures and ergodic theorem}

We can also discuss a discrete dynamical system in terms of a measure preserving transformation defined on the probability space. This approach can handle not only purely mathematical concepts but also physical phenomena in nature. This subsection is devoted to invariant measures, absolutely continuous measures and the most fundamental idea in ergodic theory --- the Birkhoff Ergodic Theorem, which states that with probability one the average of a function along an orbit of an ergodic transformation is equal to the integral of the given function.

{\bf Definitions.} Let  $(X,\mathcal{B},\mu)$ be a probability space and $f\colon X\mapsto X$ be a measurable map. The measure $\mu$ is $f$-invariant (a map $f$ is $\mu$-invariant) if $\mu(f^{-1}E)=\mu (E)$ for every $E\in \mathcal{B}$. For $f$-invariant measure $\mu$ we say that $\mu$ is ergodic ($f$ is ergodic) if $E\in\mathcal{B}$ satisfies $f^{-1}E=E$ if and only if $\mu(E)=0$ or $1$. A measure $\mu$ is absolutely continuous with respect to Lebesgue measure if and only if for every set $E\in\mathcal{B}$ of zero Lebesgue measure $\mu(E)=0$.

We can now state ergodic theorem.
\begin{theorem}[Birkhoff Ergodic Theorem]
Let $(X,\mathcal{B},\mu)$ be a probability space. If $f$ is $\mu$-invariant and $g$ is integrable, then \[\lim_{n\to\infty}\frac 1n \sum_{k=0}^{n-1}g(f^k (x))=g^*(x)\] for some $g^*\in L^1(X,\mu)$ with $g^*(f(x))=g^*(x)$ for almost every $x$. Furthermore if $f$ is ergodic, then $g^*$ is constant and \[\lim_{n\to\infty}\frac 1n \sum_{k=0}^{n-1}g(f^k (x))=\int_Xg \;d\mu\]
for almost every $x$. 
\end{theorem}

Lastly, why absolutely continuous invariant measures matter? Computer-based investigations are widely used to gain insights into the dynamics of chaotic phenomena. However, one must exercise caution in the interpretation of computer simulations. 
Often, chaotic systems exhibit multiple ergodic invariant measures \cite{gao2014summability}; it is thus important to distinguish between the measure exhibited by an actual orbit, and the measure of the orbit obtained from computer simulations, which may differ due to accumulated computational round-off errors. But if the absolutely continuous measure with respect to Lebesgue measure exists, then computer simulations will yield the measure that we expect \cite{boyarsky2012laws}. Thus, the theoretical measure and the computational measure coincide in this work. 


\section{Related Work}
\label{s:related}
The study of learning dynamics in game theory has a long history, dating back to the work of \citet{Brown1951} and \citet{Robinson1951} on fictitious play in zero-sum games, which shortly followed von Neumann's seminal work on zero-sum games \cite{Neumann1928,Neumann1944}.
A representative set of reference books are the following: \citet{Fudenberg98,Cesa06,young2004strategic,Hofbauer98,sandholm10,sergiu2013simple}.

{\bf Main precursors.} \citet{palaiopanos2017multiplicative} put forward the study of chaotic dynamics arising from Multiplicative Weights Update (MWU) learning in congestion games. They established the existence of an attracting limit cycle of period two and of Li-Yorke chaos for MWU dynamics in {\it atomic} congestion games with two agents and two links with linear cost functions. Symmetry of the game (i.e., the existence of a symmetric equilibrium where both agents select each path with probability $0.5$) results in a limit cycle of period two. They also showed for a {\it specific instance} of a game with an asymmetric equilibrium that MWU leads to Li-Yorke chaos, provided that agents adapt the strategies with a sufficiently large learning rate (step size) $\epsilon$ (equivalently, if agents use a fixed learning rate $\epsilon$ but their costs are scaled up sufficiently large).  
Shortly afterwards, \citet{CFMP} established that Li-Yorke chaos is prevalent in {\it any} two-agent {\it atomic} congestion games with two parallel links and linear cost functions, provided the equilibrium is asymmetric. Namely, in {\it any} $2\times2$ congestion game with an asymmetric equilibrium, Li-Yorke chaos emerges as the cost functions grow sufficiently large, but only if the initial condition is symmetric, i.e., both agents start with the same initial conditions. Furthermore, \cite{CFMP} established for the first time that, despite periodic or chaotic behaviors, the time-average strategies of both agents {\it always} converge {\it exactly} to the interior Nash equilibrium. While our current work leverages techniques from \cite{CFMP}, it also investigates other definitions of chaos, e.g., positive topological entropy, studies {\it non-atomic} congestion games, and relates the results to the Price of Anarchy and system efficiency analysis. Moreover, whereas in \cite{palaiopanos2017multiplicative,CFMP}  chaotic behavior is contained in a one-dimensional invariant subspace of the two dimensional space, in this paper the dimensionality of the system is already equal to one and hence the chaotic results are relevant for the whole statespace.
 Lastly, in the appendices, we provide preliminary results for learning dynamics in larger and more complex congestion games with many degrees of freedom.

{\bf Chaos in game theory.} Under the assumption of perfect rationality, it is not surprising that Nash equilibria are central concepts in game theory. However, in reality, players do not typically play a game following a Nash equilibrium strategy. The seminal work of \citet{SatoFarmer_PNAS} showed analytically by computing the Lyapunov exponents of the system that even in a simple two-player game of rock-paper-scissor, replicator dynamics (the continuous-time analogue of MWU) can lead to chaos, rendering the equilibrium strategy inaccessible. 
 For two-player games with a large number of available strategies (complicated games), \citet{GallaFarmer_PNAS2013} argue that experienced weighted attraction (EWA) learning, a behavioral economics model of learning dynamics, exhibits also chaotic behaviors in a large parameter space. The prevalence of these chaotic dynamics also persists in games with many players, as shown in the recent follow-up work \cite{GallaFarmer_ScientificReport18}. Thus, careful examinations suggest a complex behavioral landscape in many games (small or large) for which no single theoretical framework currently applies.
   \citet{VANSTRIEN2008259} and \citet{VANSTRIEN2011262} prove that fictitious play learning dynamics for a class of 3x3 games, including the Shapley's game and zero-sum dynamics, possess rich periodic and chaotic behavior. \citet{CP2019} prove that many online learning algorithms, including MWU, with a constant step size is Lyapunov chaotic when applied to zero-sum games. Finally, \citet{2017arXiv170109043P} has established experimentally that a variant of reinforcement learning, known as Experience-Weighted Attraction (EWA), leads to limit cycles and high-dimensional chaos
 in two agent games with negatively correlated payoffs. This is strongly suggestive that chaotic, non-equilibrium results can be further generalized for other variants of zero-sum games.

{\bf Other recent non-equilibrium phenomena in game theory.}
In recent years, the (algorithmic) game theory community has produced several  non-equilibrium results. 
\citet{daskalakis10} showed that MWU does not converge even in a time-average sense in a specific $3\times3$ game.
 \citet{paperics11} established non-convergent dynamics for replicator dynamics in a $2\times2\times2$ game and show as a result that the system social welfare converges to states that dominate all Nash equilibria. 
 \citet{ostrovski2013payoff} analyzed continuous-time fictitious play in a $3\times3$ game and showed similarly that the dynamics dominate in performance Nash equilibria.
Our results add a new chapter in this direction, providing detailed understanding of the non-equilibrium phenomena arising from MWU in {\it non-atomic} congestion games, as well as their important implications on regret and social costs. 
 
 In evolutionary game theory contexts, which typically study continuous-time variant of MWU (replicator dynamics), numerous non-convergence results are known but again are commonly restricted to small games \cite{sandholm10}. 
  \citet{piliouras2014optimization} and \citet{PiliourasAAMAS2014} showed that replicator dynamics in (network) zero-sum games exhibit a specific type of repetitive behavior, known as Poincar\'{e} recurrence.  
  Recently, \citet{GeorgiosSODA18} proved that Poincar\'{e} recurrence also shows up in a more general class of continuous-time dynamics known as Follow-the-Regularized-Leader (FTRL).  \citet{2017arXiv171011249M} established that the recurrence results for replicator extend to dynamically evolving zero-sum games.
  Perfectly periodic (i.e., cyclic) behavior  for replicator dynamics may arise in team competition \cite{DBLP:journals/corr/abs-1711-06879} as well as in network competition \cite{nagarajan2018three}. 
Works in this category combine  arguments such as volume preservation and the existence of constants of motions (``conservation of energy") to show cyclic or recurrent behaviors.  \citet{pangallo2019best} established empirically the emergence of cycles and more generally non-equilibrium behavior in numerous learning dynamics in games and showed correlations between their behavior and the behavior of much simpler best-response dynamics. Some formal connections between limit behaviors of complex learning dynamics and better response dynamics are developed in \cite{Entropy18}.

{\bf Game dynamics as physics.} 
Recently, \citet{BaileyAAMAS19} established a robust connection between game theory, online optimization and a ubiquitous class of systems in classical physics known as Hamiltonian dynamics, which naturally exhibit conservation laws.  
 In the case of discrete-time dynamics, such as MWU or gradient descent, the system trajectories are first order approximations of the continuous dynamics; conservation laws as well as recurrence then no longer hold. Instead  we get ``energy" increase and divergence to the boundary, as shown by \citet{BaileyEC18}, as well as volume expansion and Lyapunov chaos in zero-sum games, as shown by \citet{CP2019}. Despite this divergent, chaotic behavior, gradient descent with fixed step size, has vanishing regret in zero-sum games \cite{2019arXiv190504532B}. 
  So far, it is not clear to what extent the connections with Hamiltonian dynamics can be generalized; however, \citet{ostrovski2011piecewise} have considered a class of piecewise affine Hamiltonian vector fields whose orbits are piecewise straight lines and developed the connections with best-reply dynamics.  The connection between game theory and physics can hopefully enable us to understand and possibly exploit the hidden structure in non-equilibrium game dynamics, similarly to how in this paper we showed formally that chaotic dynamics have their time-average converging to the value equal to the equilibrium.


{\bf Game dynamics as dynamical systems.}
Finally, \citet{Entropy18,papadimitriou2019game} put forward a program for linking game theory to topology of dynamical systems, specifically to Conley's fundamental theorem of dynamical systems \cite{conley1978isolated}. 
 This approach shifts attention from Nash equilibria to a more general notion of recurrence, called chain recurrence. This notion generalizes both periodicity and Poincar\'{e} recurrence and as such can express the above results in a single framework. Whether and to what extent this framework will become useful depends on numerous factors, including the possibility of successfully incorporating it into a computational, experimental framework (see \cite{omidshafiei2019alpha} for a current approach). Note that our paper also attempts to form a bridge to the dynamical systems literature, especially to the richly developed theory of interval maps as well as to ergodic theory.


\section{Discussion on fixed versus shrinking step sizes and regret}
\label{s:discussion}
{\bf What about shrinking step sizes and vanishing regret?}
In this paper, we examine MWU with a fixed step size $\epsilon$ that has  non-vanishing regret. Can our results be disregarded if the agents leverage shrinking step sizes (that depend on the length of the game history), e.g. $\epsilon=1/\sqrt{T}$, which results in a vanishing regret of $O(1/\sqrt{T})$? Is applying shrinking step sizes  a quick and painless fix? \textit{The answer is no. The reason lies inside the big O notation.}

{\bf $O(1/\sqrt{T})$-regret: What is inside the O?}
When agents implement MWU with shrinking step-size $\epsilon(T)=1/\sqrt{T}$  and the cost in each time step is \( c_n : \mathcal{A} \rightarrow \mathbb{R} \), with \( c_n(s) \in [0,M] \), then its regret is  
\[
 \underbrace{\sum_{t=1}^T \E_{a_n \sim x_n} c_n(a_n)}_{\text{MWU with $1/\sqrt{T}$ step size}} < \underbrace{\min_{a \in A} \sum_{n=1}^T c_n(a)}_{\text{best fixed action}} + (M+1) \sqrt{T \log(|\mathcal{A}|)}
 \]
where $|\mathcal{A}|$ is the number of strategies available to the agents (see also  \cite{Cesa06} [Sections 2.6, 2.8, Remark 2.2] for the discussion on why the term \( O(M\sqrt{T\log(|\mathcal{A}|)}) \) cannot be further improved in general optimization settings). Hence, the time-average regret is $\frac{(M+1) \sqrt{\log(|\mathcal{A}|)}}{\sqrt{T}}$, which vanishes as $T\rightarrow \infty$. However, for large enough $M$ the amount of time $T$ for the regret to become negligibly small can be impractically large.
 In the case of games, due to the stability of the online payoff streams one can prove stronger regret bounds \cite{Syrgkanis:2015:FCR:2969442.2969573,foster2016learning} including $\Theta(1/T)$ \cite{GeorgiosSODA18}  for all (continuous-time) Follow-the-Regularized-Leader (FTRL) dynamics, which encompasses MWU. However, these bounds imply that in order to reach a state of small regret $\epsilon$, we still require a number of steps that is polynomial in $M/\epsilon$, where $M$ is the largest possible cost value in our game.

What is the value of $M$ in our setting of congestion games? It is the worst  possible cost $M=N\max\{ \alpha, \beta\}$.   So, for a large population size $N$, even for MWU with shrinking $\epsilon$, the wait until the regret is negligibly small can be impractically long. For any meaningful time horizon, the regret of the agents can still be so large that the $(\lambda,\mu)$-robustness type of results \cite{Roughgarden09} cannot be applied. A new theoretical framework to study these long transient periods with large regrets is needed. 

{\bf The cost normalization ``trick" only masks the problem of slow convergence.} In any game, including congestion games with many agents, one can normalize the costs so that they lie in $[0,1]$ instead of $[0,M]$ and indeed this is the standard practice in the Price of Anarchy literature when analyzing no-regret dynamics. In this case, the regret term appears more innocuous: 

\[
 \underbrace{\sum_{t=1}^T \E_{a_n \sim x_n} c_n(a_n)}_{\text{MWU with $1/\sqrt{T}$ step size}} < \underbrace{\min_{a \in A} \sum_{n=1}^T c_n(a)}_{\text{best fixed action}}\ \ \ \ \ + \underbrace{2\sqrt{T \log(|\mathcal{A}|)}}_{\text{regret in the unit of worst possible cost}}
 \]

Of course, this does not fix the problem as the regret term only looks smaller but in fact it is expressed in really large units. This cost normalization suggests that if e.g. in NYC all drivers use the same road, then the cost experienced by them is equal to $1$. 
 Naturally this nightmarish scenario, if it could ever be enforced, would translate into a monstrous traffic jam that would require hundreds or thousands of hours to resolve. So, the real regret in this setting is $2\sqrt{T \log(|\mathcal{A}|)}\times$(the duration of the worst traffic jam possible in a city of millions of people). This is an enormous number and can only be amortized by running the system for e.g. hundreds of years in the case of NYC. The usage of measuring units that scale up with the problem that is being measured (like a measuring tape that keeps adjusting the notion of a meter) creates quite a bit of confusion about how efficient are the states reachable within a reasonable time horizon. In this work, we show that the effective scaling up of the cost functions due to the increase in total demand/population can lead to a qualitatively different behavior (chaos instead of equilibration) hence these effects cannot be safely discounted by normalizing costs but instead need to be carefully studied.

{\bf The slow convergence of no-regret algorithms to correlated equilibria is well supported by theoretical work in the area.}
No-regret algorithms are rather inefficient in finding coarse correlated equilibria in large games, exactly because their regrets vanish at a slow rate. This is why different centralized ellipsoid methods algorithms have been developed to compute correlated equilibria in games with many agents that necessitate compact descriptions \cite{Papadimitriou:2008:CCE:1379759.1379762,jiang2011polynomial}.  Quoting from \cite{Papadimitriou:2008:CCE:1379759.1379762}, 

\begin{quote} 
``(No-regret) learning methods require an exponential number of iterations to
converge, proportional to $(1/\epsilon)^k$ for a constant $k > 1$. Our ellipsoid-based algorithm, on the other hand is polynomial in  $\log (\frac{1}{\epsilon})$.'' \dots ``Intractability of an equilibrium concept would make it implausible as a model of behavior. In the words of Kamal Jain: ``If your PC cannot find it, then neither can the market.""
\end{quote}


Our analysis explains the algorithmic behavior along this  {\it long} transient  (metastable) epoch, showing that the behavior can be rather different and much more inefficient than the standard asymptotic equilibrium analysis suggests. 
 In fact, computational experiments on real-life congestion games such as in wireless networks appear to be in agreement with our theoretical analysis of 
 slow/non-stabilization and bad performance of MWU and variants, which we now discuss.

{\bf MWU and its variants fail to converge efficiently in real-world applications of congestion settings.}  \citet{appavoo2018shrewd} studies, via simulations, resource selection problems for mobile networks, which are formulated as congestion games. They study the performance of MWU and  EXP3, a well known multi-armed bandit variant of MWU, and show that even in relatively small instances (20 devices/agents and 3 networks/resources), these algorithms fail to equilibrate, and in fact  perform worse than naive greedy solutions.  \cite{2019arXiv190107768M}  iterates more clearly on these failures, identifying the slow rate of stabilization as the main culprit behind the poor performance. 

\begin{quote}
``\dots we consider the problem of wireless network selection. Mobile devices are often required to choose the right network to associate with for optimal performance, which is non-trivial. The excellent theoretical properties of EXP3, a leading multi-armed bandit algorithm, suggest that it should work well for this type of problem. Yet, it performs poorly in practice. A major limitation is its slow rate of stabilization."
\end{quote}
 
 {\bf Is a step-size of $2^{-100}$ realistic as a model of human behavior?}
A continuously shrinking learning step size is a rather artificial model of human learning behavior; a sufficiently small step size means learning almost does not happen. More plausibly, to exclude the situation with unrealistic no-learning behavior, a lower bound on the learning rate $\epsilon$ should be imposed.  It is this exact parameter that we shall adopt as the fixed learning rate of our MWU model. The fact that we are using a fixed (but arbitrarily small) instead of continuously shrinking step size with a limit value of zero, is a feature of our model that does not artificially curtail agent adaptivity merely to enable theoretical results that only become binding after unreasonably long time horizons.

{\bf Shrinking step-sizes versus increasing populations.} Last but not least, we now show how the analysis for a fixed learning rate $\epsilon$ can easily be extended to capture non-equilibrating phenomena for arbitrary sequences of shrinking step sizes, as long as we allow for a dynamically evolving, increasing population. It should be already clear that the step-size $\epsilon$ and the population size $N$ (or equivalently the value of the maximum cost $M$) are competing forces that control system's stability. The larger population size implies the larger maximum cost $M$, which in turn implies the larger time horizon for MWU with a shrinking step-size algorithm to acquire smaller time-average regret, and for the classic equilibrium, Price of Anarchy, analysis to restore its predictive power. Unfortunately, if the population increases at a sufficiently fast rate to counter the shrinking step-size rate, the time-average regret will never vanish. Specifically, from our analysis, we proved that the relevant parameter that controls the long-time dynamics (e.g. equilibration, limit cycles, or chaos) and the social cost is 
$a=(\alpha+\beta) N \ln\left(\frac1{1-\eps}\right)$, see Section \ref{var}. As long as at every time step $n$,  $a(n)=(\alpha+\beta) N(n) \ln\left(\frac1{1-\eps(n)}\right)$ is greater than the chaotic threshold then the system will always remain in the chaotic regime despite the step-size going to zero. For example, for $\eps(n)=1/\sqrt{n}$, it suffices that $N\geq \frac{a_b}{(\alpha+\beta) \ln\left(\frac1{1-1/\sqrt{n}}\right)}$ where $a_b$ is the chaotic threshold defined in Theorem \ref{per3}.
Simple calculations show that it suffices $N\geq \frac{a_b}{(\alpha+\beta)}  \sqrt{n} \geq \frac{a_b}{(\alpha+\beta) \ln\left(\frac1{1-1/\sqrt{n}}\right)} $, 
 Namely, a slowly (sublinearly) increasing population suffices  for the system to remain forever in its non-equilibrating, inefficient,  chaotic regime.  


\section{Conclusion}
\label{s:conclusion}

We explore the Multiplicative Weight Update algorithm in two-strategy non-atomic congestion games. We find that standard game-theoretic equilibrium analysis, such as the Price of Anarchy, fails to capture extremely rich non-equilibrium phenomena of our simple model. Even when the Price of Anarchy is equal to $1$, the system can be dynamically unstable when the total demand $N$ increases. 
Every system has a carrying capacity, above which the dynamics become non-equilibrating. In fact, they become chaotic when $N$ is sufficiently large, provided the equilibrium flow is asymmetric.

 This demand-driven instability is a robust phenomenon that holds for many different setting of congestion games (many paths, atomic/non-atomic congestion games, different cost functions, etc).
    In the case of linear cost functions the dynamics we study here also exhibit a remarkable time-average property; namely, in any non-equilibrating regime driven by large total demand, the time-average flows/costs of the paths converge {\it exactly} to the Nash equilibrium value, a property reminiscent of the behavior of regret minimizing dynamics in zero-sum games.
    In the case of polynomial cost functions, the time-average costs exhibit a similar regularity in the sense that no individual path appears significantly cheaper on average.
     Interestingly, when we keep increasing the total system demand, congestion games eventually ``break down" and flip their characteristics to become more like zero-sum games.  On the other hand,
    the time-average convergence property {\it does not} guarantee small regret nor low social costs, even when the Price of Anarchy is equal to $1$. In fact, time-average regret and time-average social costs increase with fluctuations from the equilibrium value, which can be maximally large in the non-equilibrating regime. In the case of a symmetric equilibrium flow, fluctuations arise from an extreme swing; almost all users will take the same route and simultaneously alternate between the two routes, when the population size is large. The time-average social cost in this situation can thus be as high as it can get.

Our benign-looking learning in games model is full of surprises and puzzles. The dynamical system approach provides a useful framework to investigate the unusual connections between non-equilibrating dynamics, and the classic game-theoretic (equilibrium) metrics such as regret, and Price of Anarchy. 
 For instance, we show in the Appendix that, in certain non-equilibrating regimes, the system despite having a unique equilibrium may have multiple distinct attractors and hence the time-average regret and social cost depend critically on initial conditions. Also in the Appendix,
 we report other interesting observations, discuss future directions and even include extensions to numerous settings.
   Notable properties of periodic orbits, such as the coexistence of two attracting periodic orbits, and the Feigenbaum's period-doubling bifurcation route to chaos are also presented.

\section*{Acknowledgements}

Thiparat Chotibut and Georgios Piliouras acknowledge SUTD grant SRG ESD 2015 097, MOE AcRF Tier 2 Grant 2016-T2-1-170,  grant PIE-SGP-AI-2018-01 and NRF 2018 Fellowship NRF-NRFF2018-07.
Fryderyk Falniowski acknowledges the support of the
National Science Centre, Poland, grant 2016/21/D/HS4/01798 
and COST Action CA16228 ``European Network for Game Theory''.
Research of Micha{\l} Misiurewicz was partially supported by grant
number 426602 from the Simons Foundation.

\bibliographystyle{ACM-Reference-Format} 
\bibliography{ms} 


\appendix
\label{s:appendix}
\newpage
\begin{center}
{\large {\bf Appendices }}	
\end{center}

\section{Properties of attracting orbits}
\label{s:properties_orbits}

In this section, we investigate the properties of the attracting periodic orbits associated with the interval map $f_{a,b}: [0,1] \rightarrow [0,1]$
\begin{equation}
\label{eqn: 1dmap}
f_{a,b}(x) = \frac{x}{x + (1-x)\exp\left(a(x-b) \right)}.
\end{equation} 

We've argued in the main text that, when $b = 0.5$, the dynamics will converge toward the fixed point $b = 0.5$ whenever $a < 8$. And for any $a \ge 8$, the long-time dynamics will converge toward the attracting periodic orbits of period 2 located at $\{ \sigma_a, 1 - \sigma_a \}$. The bifurcation diagram is thus symmetric around $b = 0.5$ as shown in the top picture of Fig. \ref{fig: regbound_b0p5}. In this case, the time-average regret is well-approximated by its upper bound, and the normalized time-average social cost asymptotes to the maximum value of $2$.

\begin{figure}[hbt!]
  \centering
   {\includegraphics[width = 0.9\textwidth]{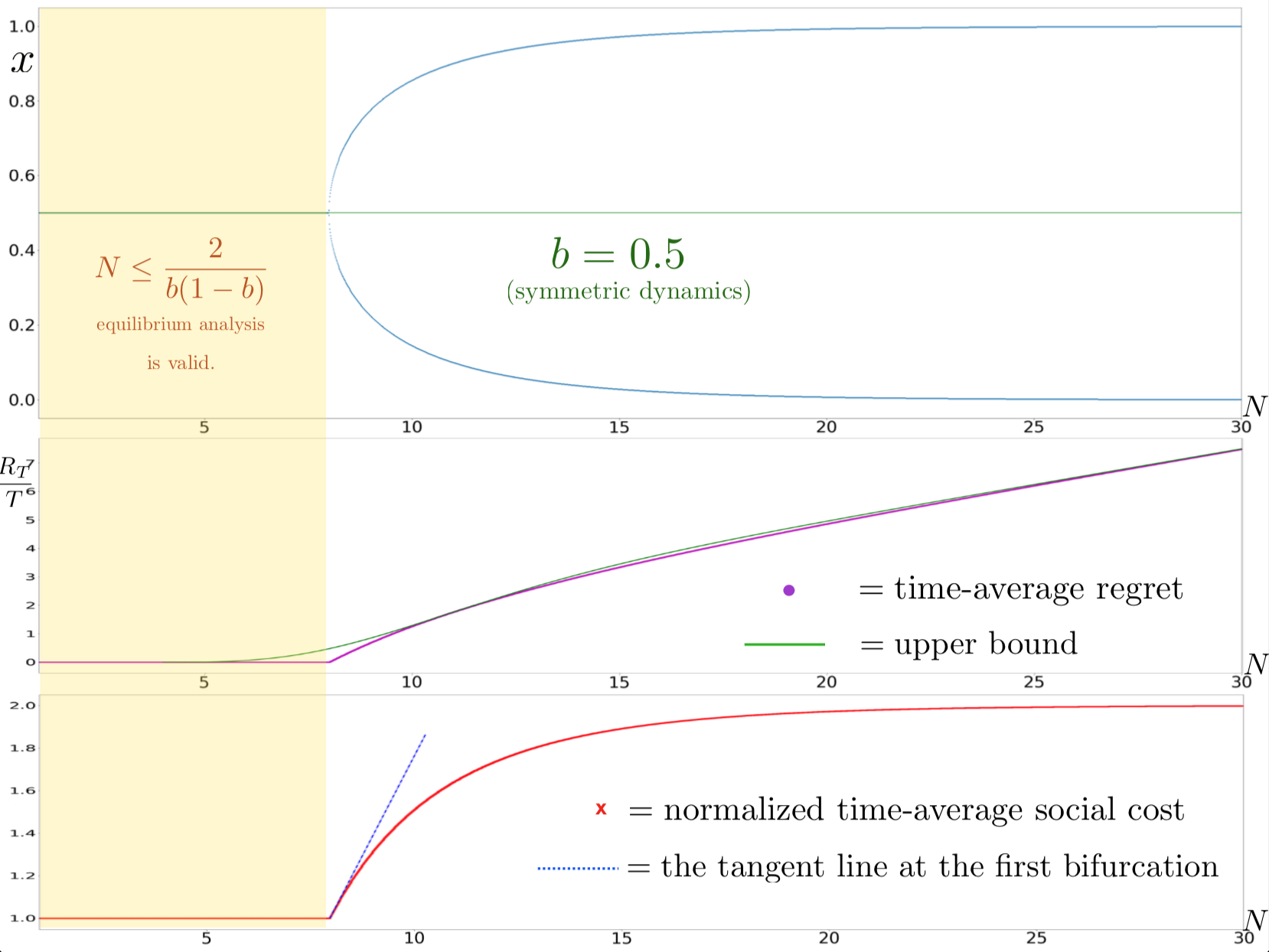}}
   \caption{ Even in the symmetric case of $b=0.5$, bifurcation diagram (Top) demonstrates instability of the routing game driven by the increase in total demand $N.$ Here, $\epsilon = 1 - 1/e$ so $a = N$ as usual. In this symmetric case, the capacity of the network, under which long-time dynamics equilibrate, is $N^*_b = 8$. Above the capacity, attracting periodic orbits of period 2 emerges.  (Middle) The time-average regret, shown in the purple circle symbol, suddenly becomes strictly positive at the bifurcation. The regret bound also well approximates the actual values. (Bottom) normalized time-average social cost also {\it suddenly} becomes greater than 1 at the bifurcation. Even in the symmetric case, the classic Price of Anarchy metric fails for $N > N^*_b = 8$.}
  \label{fig: regbound_b0p5}
\end{figure}

When $b$ differs from $0.5$, we have argued in the main text that the emergence of chaos is inevitable, provided $a$ is sufficiently large. The period-doubling bifurcations route to chaos is guaranteed to arise. Fig. \ref{fig: intro_summary} of the main text shows chaotic bifurcation diagrams when   $b = 0.7.$ In this asymmetric case, standard equilibrium analysis only applies when the fixed point $b$ is stable, which is when $ |f'_{a,b}(b)| \le 1$, or equivalently when $ a \le 2/b(1-b)$.

\begin{figure}[hbt!]
  \centering
   {\includegraphics[width = 0.8\textwidth]{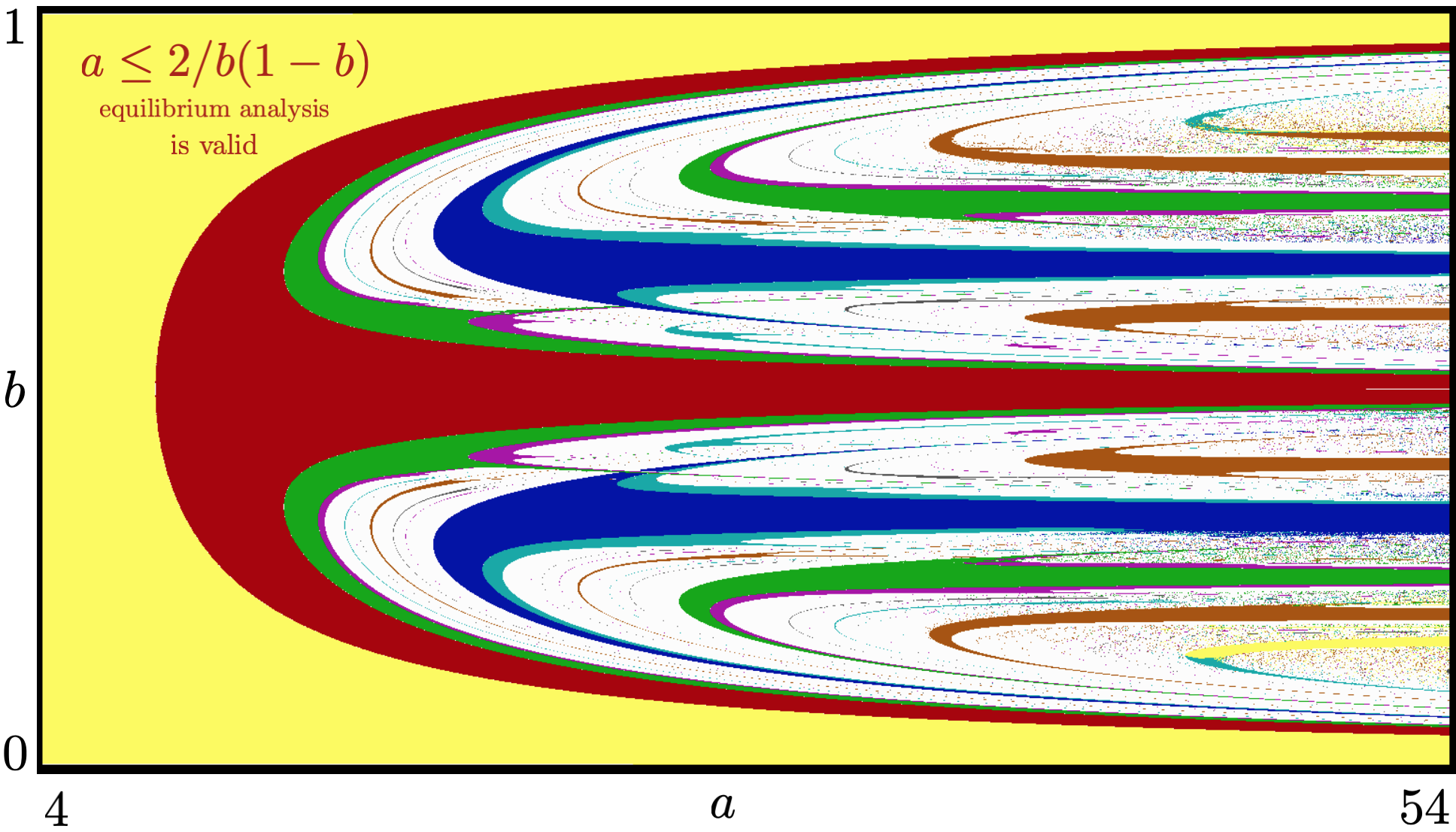}}
   \caption{Period diagrams of the small-period attracting periodic orbits associated with the map (\ref{eqn: 1dmap}). The colors encode the periods of attracting periodic orbits as follows: period 1 (fixed point) = {\color{yellow}yellow}, period 2 = {\color{red}red}, period 3 = {\color{blue}blue}, period 4 = {\color{green}green}, period 5 = {\color{brown}brown}, period 6 = {\color{cyan}cyan}, period 7 = {\color{darkgray}darkgray}, period 8 = {\color{magenta}magenta}, and period larger than 8 = white. The equilibrium analysis is only viable when the fixed point $b$ is stable, i.e. when $a \le 2/b(1-b)$. In other region of the phase-space, non-equilibrating dynamics arise and system proceeds through the period-doubling bifurcation route to chaos in the white region. The picture is generated from the following algorithm: 20000 preliminary iterations are discarded. Then a point is considered periodic of period $n$ if $|f^n(x)-x|<0.0000000001$ and it is not periodic of any period smaller than $n$. Slight asymmetry is caused by the fact that the starting point is the left critical point $x_l = 1/2 - \sqrt{1/4-1/a}$. In addition, for a fixed $a$, as we vary $b$ and penetrate into the chaotic regimes (white) from the outer layers, we numerically observe Feigenbaum's universal route to chaos as discussed below.}
  \label{fig: periods_color}
\end{figure}

{\bf Feigenbaum's universal route to chaos}:
The period diagrams as a function of the two free-parameters $a$ and $b$ are shown in Fig. \ref{fig: periods_color}. It's interesting to report numerical observations of Feigenbaum's route to chaos for our bimodal map $f_{a,b}$. Although Feigenbaum's universality is known to apply among a one-dimensional {\it unimodal} interval map with a quadratic maximum \cite{feig, landford, tabor}, we also observe the Feigenbaum's period-doubling route to chaos for our {\it bimodal} interval map. Specifically, by fixing $a$ and varying $b$, we numerically measure the ratios 
\begin{equation}\label{eqn: feig_def}  
\delta_n \equiv \frac{b_{n+1} - b_n}{b_{n+2} - b_{n+1}},\ \ \ \alpha_n \equiv \frac{d_n}{d_{n+1}},
\end{equation}
where $b_n$ denotes the value at which a period $2^n$-orbits appears, and $d_n = f_{a,b}^{2^{n-1}}(x_l) - x_l$ such that the left critical point $x_l=\frac 12\left(1-\sqrt{1-\frac 4a}\right)$ (the point at which $f_{a,b}$ attains its maximum) belongs to the $2^n$-orbits\footnote{In this way, we can numerically approximate the {\it signed} second Feigenbaum constant $\alpha$ \cite{Strogatz2000}.}.
As $n$ grows large (we truncate our observation at $n = 12$), we find 
\begin{equation}\label{eqn: feig}  
\delta_{n = 12} \approx 4.669\dots, \ \ \ \alpha_{n=12} \approx -2.502\dots, \ \  \end{equation}  
which agree, to 4 digits, with the Feigenbaum's universal constants, $\delta = 4.669201609102990\dots$ and $\alpha = -2.502907875\dots$, 
that appear, for example, in the period-doubling route to chaos in the logistic map.

{\bf Coexistence of two attracting periodic orbits and non-uniqueness of regret and social cost}:
The map $f_{a,b}$ has a negative Schwartzian derivative when $a>4$,  thus it has at most two  attracting or neutral periodic orbits. Although the time-average of every periodic orbits converges {\it exactly} to the Nash equilibrium $b,$ the variance $\lim_{T \rightarrow \infty} \frac{1}{T} \sum_{n=1}^T (x_n - b)^2$ of the coexisting periodic orbits can differ. Thus, the normalized time-average social cost and the time-average regret, which depend on the variance, can be multi-valued. Which value is attained depends on the variance of the attracting periodic orbits that the dynamics asymptotically reaches, which itself depends on the initial condition $x_0$. 
Period diagrams of Fig. \ref{fig: periods_coexistence} reveal how the two coexisting initial condition-dependent attracting periodic orbits are intertwined, and  Fig. \ref{fig: nonunique_orbits} reports the evidence of two coexisting periodic orbits whose variances differ, leading to multi-valued time-average regret and social cost.

\begin{figure}[hbt!]
  \centering
   {\includegraphics[width = 0.8\textwidth]{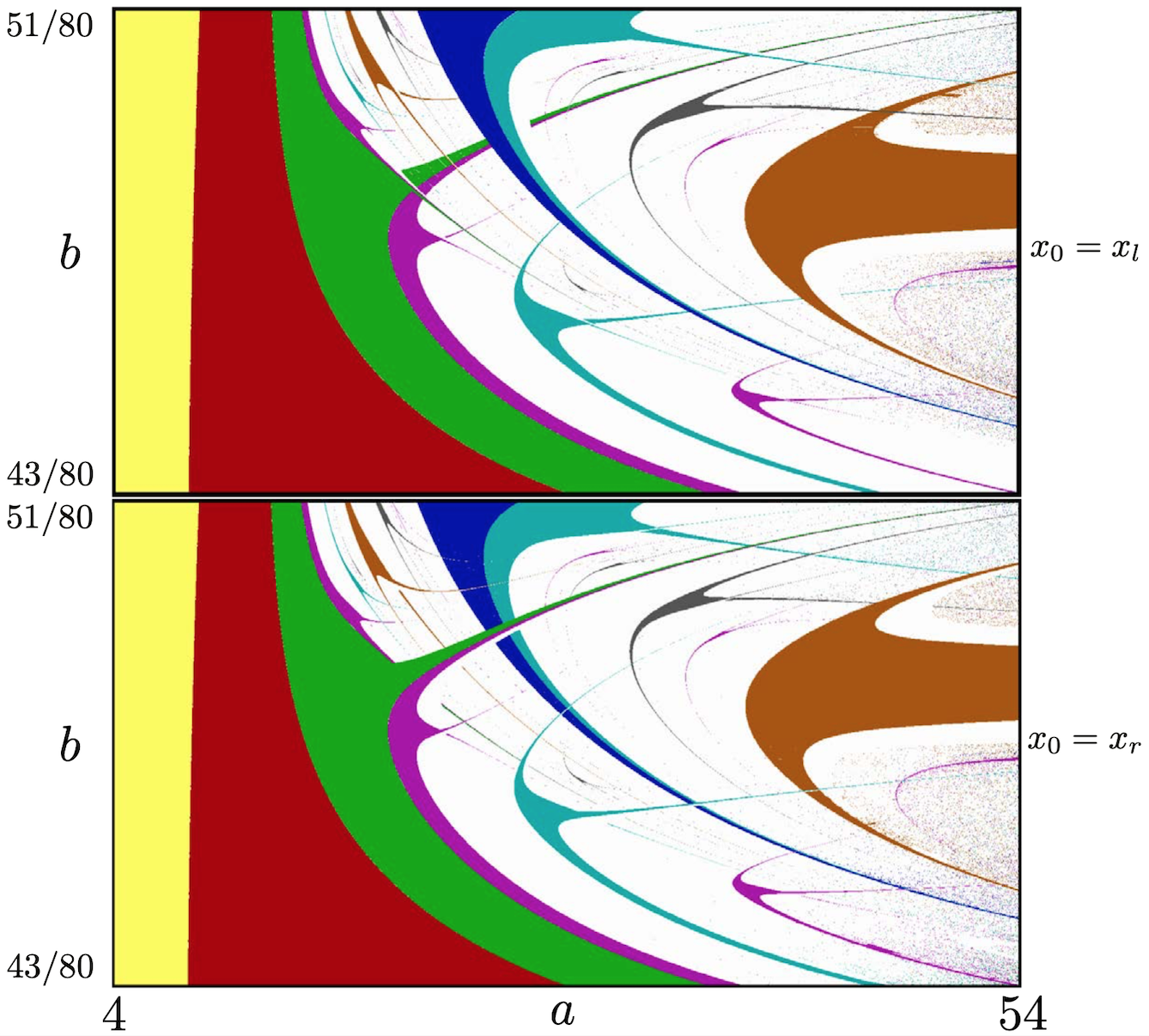}}
   \caption{Coexistence of two initial condition-dependent attracting periodic orbits. The pictures are generated from the same procedure as explained in Fig.\ref{fig: periods_color}, except that here the initial conditions for the top and the bottom pictures are located at the left and the right critical points, respectively. Also, $b \in [43/80, 51/80]$ and $a \in [4, 54].$ The color schemes are the same as those of Fig. \ref{fig: periods_color} : period 1 (fixed point) = {\color{yellow}yellow}, period 2 = {\color{red}red}, period 3 = {\color{blue}blue}, period 4 = {\color{green}green}, period 5 = {\color{brown}brown}, period 6 = {\color{cyan}cyan}, period 7 = {\color{darkgray}darkgray}, period 8 = {\color{magenta}magenta}, and period larger than 8 = white.}
  \label{fig: periods_coexistence}
\end{figure}

\begin{figure}[hbt!]
  \centering
   {\includegraphics[width = 0.9\textwidth]{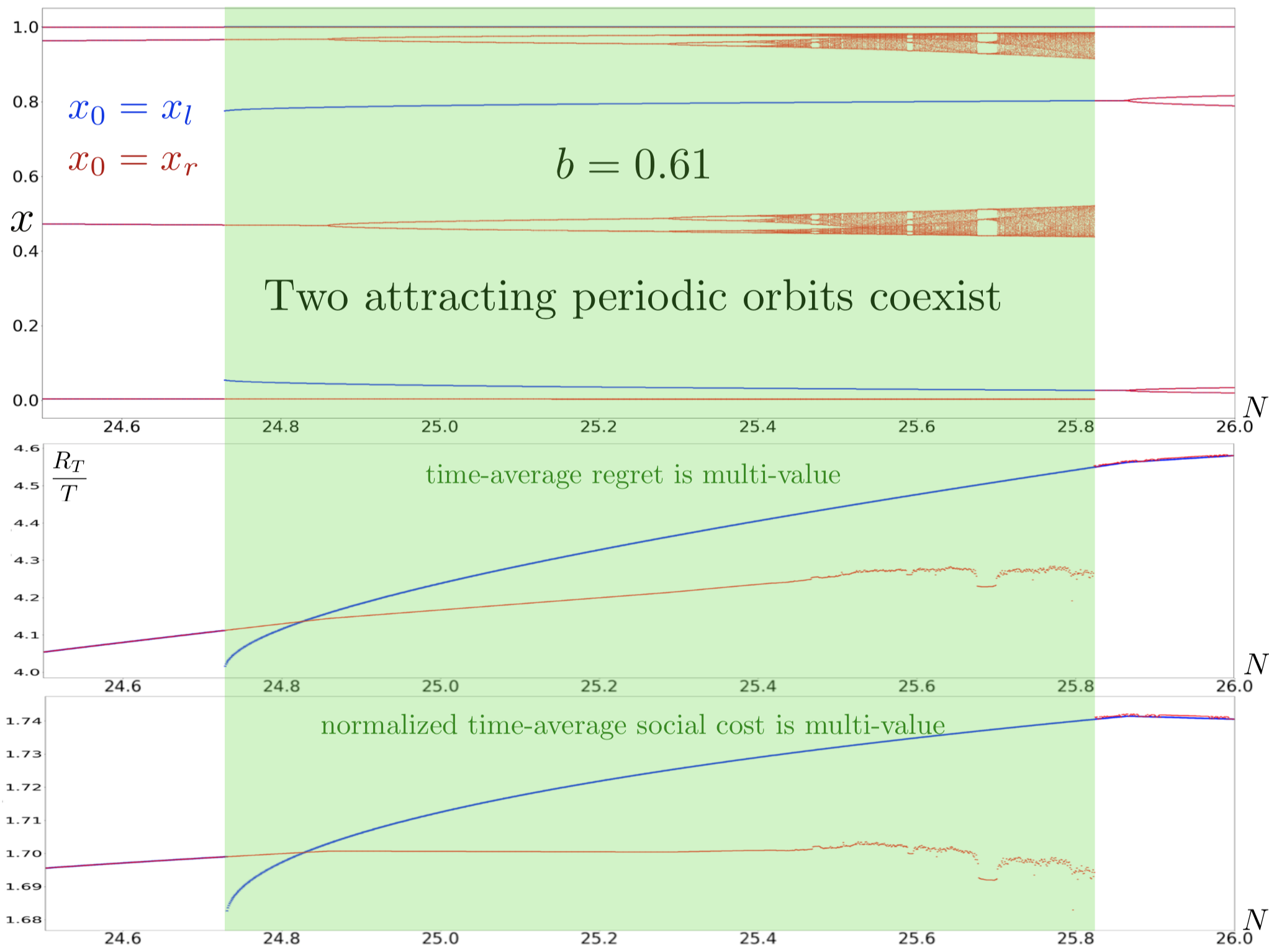}}
   \caption{Coexistence of two attracting periodic orbits at $b = 0.61$ with two different variances implies {\it non-uniqueness} of time-average regret and normalized time-average social cost. As usual, we set $\epsilon = 1 - 1/e$ so that $N = a$. (Top) The range of $N$ in the shaded green region show coexistence of two attracting periodic orbits. The blue (red) periodic orbits is selected if the initial condition is the left (right) critical point $x_l$ ($x_r$). There are at most 2 coexisting attracting periodic orbits, as guaranteed by the negative Schwartzian derivative for our bimodal map $f_{a,b}$. The variance of the two periodic orbits are clearly different; thus, the time-average regret (middle) and the normalized time-average social cost (bottom) which depend on the variance are multi-valued. Which values are attained depend on initial conditions. }
  \label{fig: nonunique_orbits}
\end{figure}

\begin{figure}[hbt!]
  \centering
   {\includegraphics[width = 0.7\textwidth]{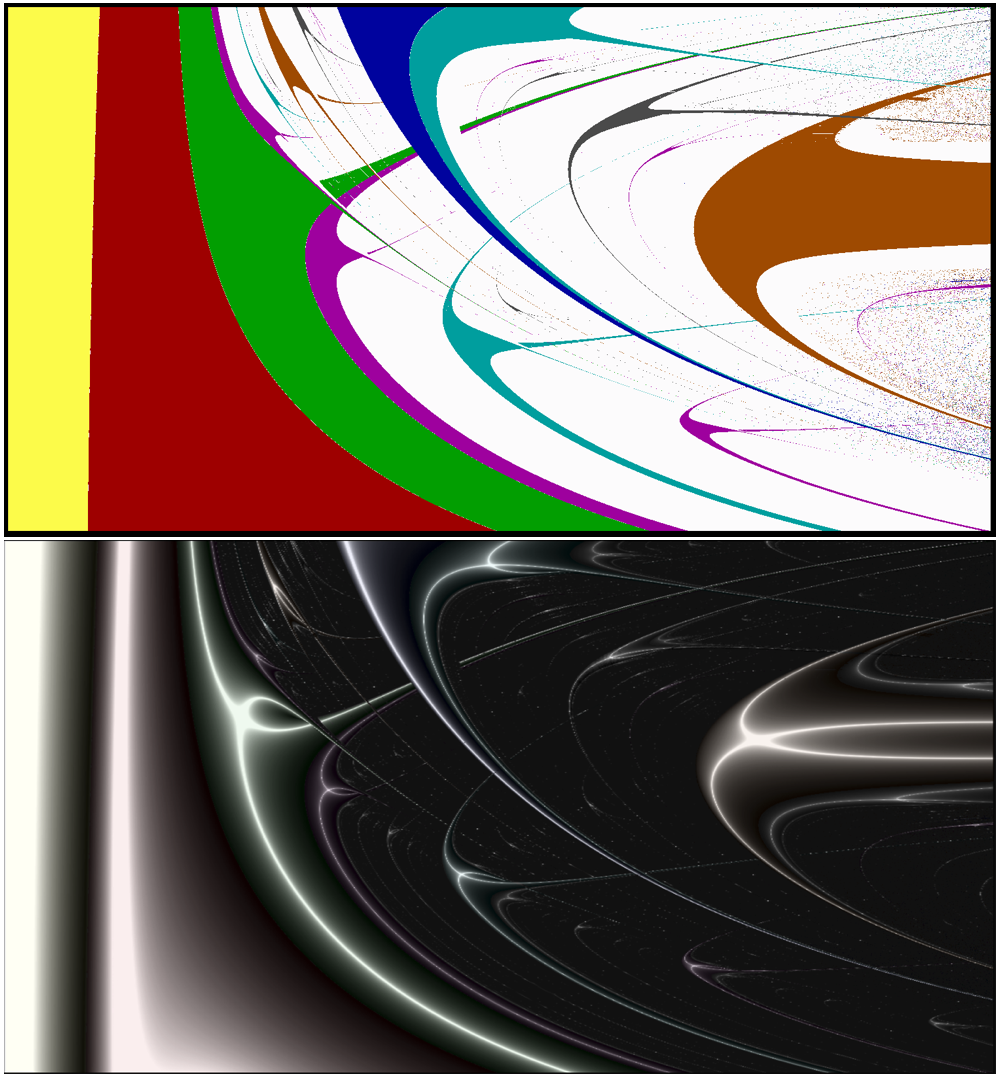}}
   \caption{(Bottom) Lyapunov exponents $\langle \log |f'_{a,b}|\rangle$ numerically approximated by $\frac{1}{T}\sum_{n=1}^T \log |f'_{a,b}(x_n)|$ with $T = 2000$, shown in gray scale, superposed on the period diagrams (Top) adopted from Fig. \ref{fig: periods_coexistence} (Top). The color scheme of the Lyapunov exponents is such that $\langle \log |f'_{a,b}|\rangle < -1.5$ is shown in white (very stable orbits) and $\langle \log |f'_{a,b}|\rangle > 0$ is shown in black (unstable or chaotic). One can clearly see that the extended-leg structures arise from having superstable orbits as the skeleton of each attracting periodic orbits regime. When both critical points are elements of the attracting orbit, the two extended legs intersect. As expected, close to the bifurcation boundaries and in the chaotic regime, the orbits becomes unstable, as represented by the black color.}
  \label{fig: lyapunov}
\end{figure}

{\bf Stability of the orbits}: 
In addition to the period diagrams, we investigate the stability of the attracting orbits by considering the Lyapunov exponents $\langle \log | f'_{a,b} | \rangle$, where $\langle \cdot \rangle$ denotes time-average. Fig. \ref{fig: lyapunov} (bottom) shows the Lyapunov exponents associated with different attracting orbits, revealing that extended-leg structures arise from the situations when the orbits become superstable, that is when one of the two critical points is an element of the orbits\footnote{Recall that the orbit is superstable if one of the critical points is an element of the orbits, so that $f'_{a,b}(x_c) = 0.$ This means the Lyapunov exponents in principle is $- \infty$, visualized as a white bright color.}. Within the regime of the same period (same color), there are situations when the two superstable extended-leg curves intersect. These scenarios happen when {\it both} critical points are elements of periodic orbits.

Also, note Fig. \ref{fig: periods_color} reveals that the qualitatively similar extended-leg structures in the period diagrams appear in layers, with a chaotic regime sandwiched between two layers. Notice also that the consecutive layers have periods differ by 1. To understand why these layers with increasing periods appear, we investigate superstable periodic orbits in these layers and found that,  all elements of the orbits, except for the left critical points $x_l = 1/2 - \sqrt{1/4 - 1/a}$ and its image $f_{a,b}(x_l)$, are approximately 0, independent of the period of the orbits. With this observation, we now approximate one of the superstable regions within each layer, using the time-average convergence to the Nash equilibrium property of Corollary \ref{cmper}. Namely, let $x_l$ be an element of a periodic orbit of period $p$ such that only $x_l$ and $f_{a,b}(x_l)$ are significantly larger than 0, then from Corollary \ref{cmper} we have 
\begin{equation}\label{eq: approx_ssorbits}
x_l + f_{a,b}(x_l) + \underbrace{\left\{ f_{a,b}^2(x_l) + \dots + f_{a,b}^{p-1}(x_l) \right\}}_{\approx \ 0}	 = pb.
\end{equation}
Numerical results show that the approximation that every elements of the periodic orbits except $x_l$ and $f_{a,b}(x_l)$ are close to 0 becomes better and better for periodic orbits with larger periods; hence, we're interested in the limit of $a \gg 1$. To leading order in $\frac{1}{a}$, $x_l \approx \frac{1}{a}$ and $f_{a,b}(x_l) \approx \frac{1}{1+a e^{1-ab}}$ so that (\ref{eq: approx_ssorbits}) gives $\frac{1}{a} + \frac{1}{1+a e^{(1-ab)}} \approx pb.$
\begin{figure}[hbt!]
  \centering
   {\includegraphics[width = 0.72\textwidth]{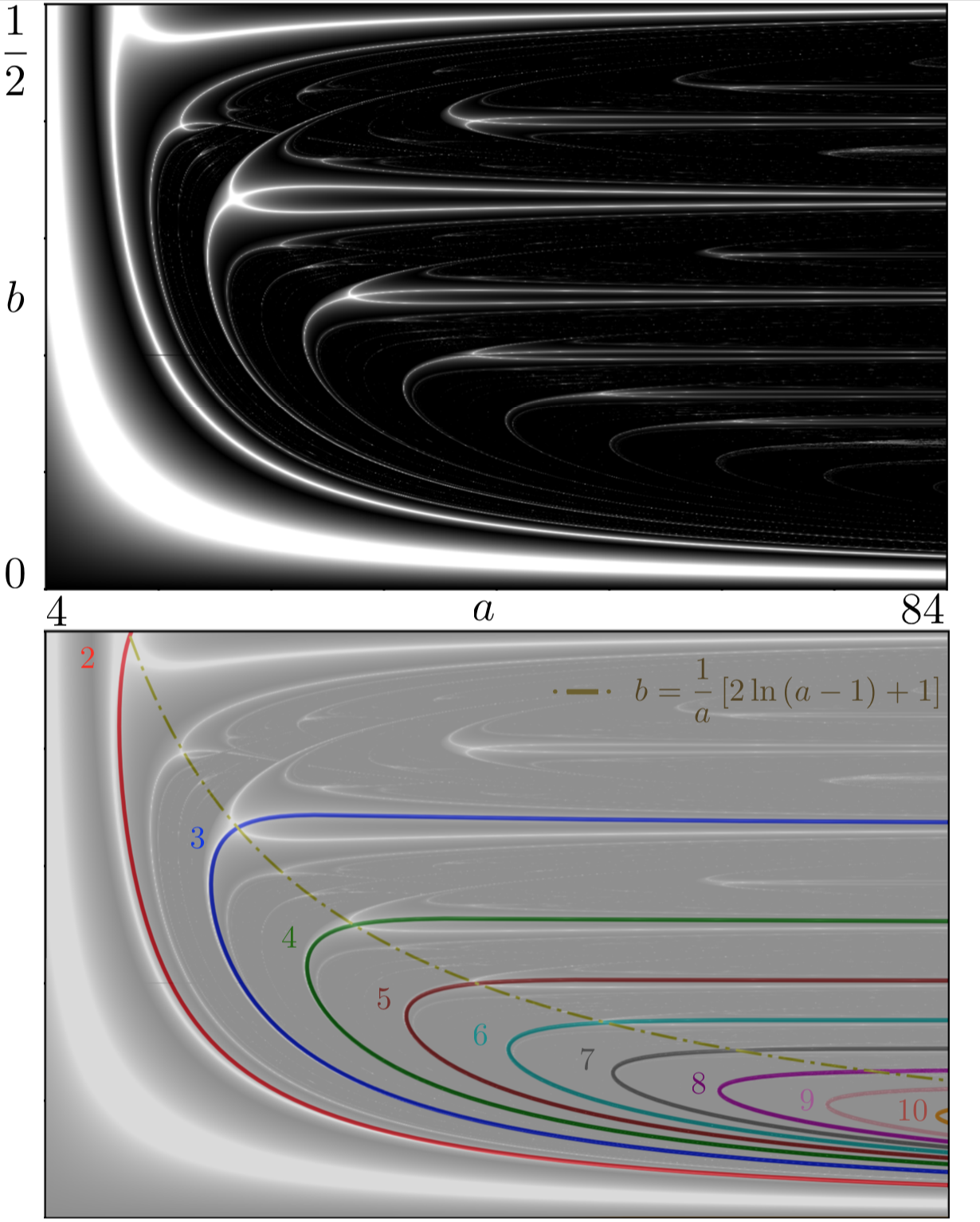}}
   \caption{Layers of extended-leg structures with increasing periods arise from specific permutations of superstable periodic orbits. As argued in the stability of the orbits section, $S(a,b) = p$ defines the superstable periodic orbits of period $p$ with the property that only $x_l$ and $f(x_l)$ are the only two elements of the periodic orbits that are not near $0$. The level sets of $S(a,b)$ at $p = 2, 3, \dots, 10$ are displayed in different colors (bottom), which accurately track extended-leg curves with large negative Lyapunov exponents (top). The color scheme of the Lyapunov exponents is such that $\langle \log |f'_{a,b}|\rangle < -1.5$ is shown in white (very stable orbits) and $\langle \log |f'_{a,b}|\rangle > 0$ is shown in black (unstable or chaotic).  The dashed olive green curve $b = \frac{1}{a} \left[ 2 \ln(a-1) + 1\right]$ obtained from (\ref{eqn: bothcps}) encompasses the situations when {\it both} critical points are the only two non-near-zero elements of the periodic orbits, i.e. when the two superstable curves in each region of the same period $p$ intersect. These results provide a reasonable answer to why layers of extended-leg structures with increasing periods appear in the period diagram of Fig. \ref{fig: periods_color}. }
  \label{fig: ssorbits_crit}
\end{figure}
Defining 
\begin{equation}
S(a,b) = \frac{1}{ab} + \frac{1}{b + (ab) e^{(1-ab)}},
\end{equation}
we obtain the condition
\begin{equation}\label{eqn: levelset}
S(a,b) \approx p, 
\end{equation}
that should become more accurate as $a\gg1$, for $x_l$ to be on the periodic orbit of period $p$ with the aforementioned property. Fig. \ref{fig: ssorbits_crit} reveals that the level sets of $S(a,b)$ for $p = 2, 3, \dots, 10$ accurately tracks the extended-leg structures  with increasing periods, showing that these superstable orbits are the skeletons of the extended-leg layers shown in Fig. \ref{fig: periods_color}.

In addition, we can approximate the condition when {\it both} critical points $x_l$ and $x_r$ become the elements of these superstable periodic orbits. In these specific permutations of the orbits, we require $x_r = f_{a,b}(x_l)$. And from (\ref{eq: approx_ssorbits}) we obtain $x_l + f_{a,b}(x_l) \approx pb$. Since $x_l + x_r = 1$, we conclude that both critical points will be on the periodic orbit of period $p$ with the aforementioned property when 

\begin{equation}\label{eqn: bothcps}
b \approx \frac{1}{p}, \ \ \ \text{and   }\ \ \ \  \frac{1}{a}\left[ 2 \ln(a-1) + 1\right] \approx \frac{1}{p},
\end{equation}
where the condition on $a$ follows from (\ref{eqn: levelset}) and $b \approx \frac{1}{p}$.
Therefore, if we plot the relationship $b = \frac{1}{a}\left[ 2 \ln(a-1) + 1\right]$, the graph will encompass the situations when both critical points are on the periodic orbits with the aforementioend property. This is illustrated by the dashed olive green line of Fig. \ref{fig: ssorbits_crit} that passes through the intersections between two superstable curves within each period-$p$ region.

\section{Extensions to congestion games with many strategies} 
\label{s:extension}

In this section we extend our results on Li-Yorke chaos and time-average convergence to Nash equilibrium to the case of many strategies.

We will consider a $m$-strategy \emph{congestion game} with a continuum of players/agents, where all of them use  \emph{multiplicative
  weights update}. Each of the players controls an infinitesimally small fraction of the flow.
  We will assume that the total flow of all the agents is equal to $N$. We will denote the fraction of the players using $i$ one of the  $m$ strategies  at time $n$ as $x_i(n)$ where $i \in \{1, \dots, m\}$.
  Intuitively, this model captures how a large population of players chooses between multiple alternative, parallel paths for going from point $A$ to point $B$.
 If a large fraction of the players choose the same
strategy, this leads to  congestion/traffic, and the cost increases. We will
assume that the cost is proportional to the \emph{load}. If we denote
by $c(i)$ the cost of any player playing 
strategy number $i$, and the coefficients of proportionality are
$\alpha_i$, then we get
\begin{equation}\label{many_cost}
\begin{aligned}
c(i)&=\alpha_i N x_i ~\forall  i \in \{1, \dots, m\}.
\end{aligned}
\end{equation}

At time $n+1$ the players know already the cost of the strategies at
time $n$ and update
their choices. Since we have a continuum of agents we will assume that 
the fractions of users using the first, second, $m$-th path are respectively equal to the probabilities $x_1(n), \dots,  x_{m}(n)$.
Once again we will update the probabilities using MWU. 
 The update rule in the case of $m$ strategies is as follows:

\begin{equation}\label{mwu2}
\begin{aligned}
x_i(n+1)&=x_i(n)\frac{(1-\eps)^{c(i)}}{\sum_{j \in \{1,\dots,m 
\}} x_j(n)(1-\eps)^{c(j)}},\\
\end{aligned}
\end{equation}

The (Nash) equilibrium flow $(b_1, \dots, b_m )$ of the congestion games is the unique flow such that the cost of all paths are equal to each other. Specifically, the equilibrium is defined by $b_i= \frac{1/\alpha_i}{\sum_{j\{1,\dots,m\}}\alpha_j}$.

The MWU dynamics introduced by \eqref{mwu2} can be interpreted as the
dynamics of the map $f$ of the simplex $\Delta=\{(x_1,\dots,x_m):x_i\ge
0,\ \sum_{j=1}^m x_j=1\}$ to itself, given by
\begin{equation} \label{f1m}
f(x_1,\dots,x_m)=\left(\frac{y_1}Y,\dots,\frac{y_m}Y\right),
\end{equation}
where
\[
y_i=x_i\exp(-a_ix_i)\ \,  \textrm{with}\   a_i=N \alpha_i \ln\big(\frac{1}{1-\epsilon}\big)\ \  \textrm{and}\ \ Y=\sum_{j=1}^m y_j.
\]
We set $a_i=Np_i$, $p_i=\alpha_i \ln (1/(1-\varepsilon))$, and see what happens as $N\to\infty$.
We will show that 
if $N$ is sufficiently large then $f$ is Li-Yorke chaotic, except when
$m=2$ and $p_1=p_2$, that is $\alpha_1=\alpha_2$.

{

\subsection{Proof of the existence of Li-Yorke chaos}

In this section we will provide a generalization of Theorem \ref{per3} showing that even in congestion games with many strategies, increasing the population/flow will result to instability and chaos. Thus, the emergence of chaos is robust both in the size of the game (holds in games with few as well as many paths/strategies) as well as to the actual cost functions on the edges (chaos emerges for effectively any tuple of linear cost functions).

\begin{theorem}
\label{t:chaos_many}
Given any non-atomic congestion game with $m$ actions as described by model (\ref{many_cost}),(\ref{mwu2}), 
except for the case\footnote{Given that the case  $m=2$ with $\alpha_1=\alpha_2$ is analyzed in Theorem \ref{trajf} (emergence of a periodic orbit of period $2$) we have a complete understanding of all cases.} of $m=2$ with $\alpha_1=\alpha_2$, then there exists a total system demand $N_0$ such that for if $N\geq N_0$ the system has periodic orbits of all periods, positive topological entropy and is Li-Yorke chaotic. \end{theorem}

\begin{proof}

Set
\[
p=\frac1{\sum_{k=1}^{m-1}\frac1{p_k}}
\]
and consider the segment
\[
I=\left\{(x_1,\dots,x_m)\in\Delta:x_i=\frac{px}{p_i}\ \ \textrm{for}\ \ i<m,
\ 0\le x\le 1\right\}.
\]
We have
\[
x_m=1-\sum_{k=1}^{m-1}x_k=1-x,
\]
so indeed, $I\subset\Delta$.

We have $y_i=\frac{px}{p_i}\exp(-Npx)$ for $i<m$ and
\[
y_m=(1-x)\exp(-Np_m(1-x)),
\]
so
\[
Y=x\exp(-Npx)+(1-x)\exp(-Np_m(1-x)).
\]
Therefore, for $i<m$ we get
\[\begin{split}
\frac{y_i}Y&=\frac{\frac{px}{p_i}\exp(-Npx)}{x\exp(-Npx)+
  (1-x)\exp(-Np_m(1-x))}\\
&=\frac{p}{p_i}\ \cdot\ \frac{x}{x+(1-x)\exp(-Np_m(1-x)+Npx)}.
\end{split}\]

Thus, $f(I)\subset I$, and the map $f$ on $I$ (in the variable $x$) is
given by the formula
\[
x\mapsto \frac{x}{x+(1-x)\exp(-Np_m(1-x)+Npx)}.
\]
This formula can be rewritten as
\[
x\mapsto \frac{x}{x+(1-x)\exp\left(N(p+p_m)
\left(x-\frac{p_m}{p+p_m}\right)\right)},
\]
and we already know that this map is Li-Yorke chaotic and has positive
topological entropy and periodic orbits of all possible periods for sufficiently large $N$, provided
$\frac{p_m}{p+p_m}\ne\frac12$. Therefore in this case $f$ is Li-Yorke
chaotic and has positive topological entropy on the whole $\Delta$ for
sufficiently large $N$.

Let us now investigate the exceptional case
$\frac{p_m}{p+p_m}=\frac12$. Then $p_m=p$, so
\[
\frac1{p_m}=\sum_{k=1}^{m-1}\frac1{p_k}.
\]
Therefore,
\[
\frac2{p_m}=\sum_{k=1}^m\frac1{p_k}.
\]
However, our choice of $m$ as a special index was arbitrary, so the
only case when we do not get Li-Yorke chaos and positive topological
entropy is when
\[
\frac2{p_i}=\sum_{k=1}^m\frac1{p_k}
\]
for every $i=1,2,\dots,m$. In this case, $p_1=p_2=\dots=p_m$, so we
get $\frac2{p_i}=\frac{m}{p_i}$, so $m=2$ and $p_1=p_2$. Eventually $p_1=p_2$ if and only if $\alpha_1=\alpha_2$.
\end{proof}


\subsection{Time average convergence to equilibrium}
\label{s:average}

The goal in this section is to generalize Theorem \ref{t:Cesaro} about the time-average of the flow converging to the (Nash) equilibrium flow.  We will start with a technical lemma showing that the all (interior) initial conditions converge to an interior invariant set.
%
We consider the map $f$ given by \eqref{f1m}.
%
If $x=(x_1,\dots,x_m)\in\Delta$, then we write
$\xi(x)=\min(x_1,\dots,x_m)$.

\begin{lemma}
\label{lem:interior}
If $x\in\Delta$ and $\xi(x)>0$ then
\begin{equation}\label{e1}
\inf_{n\ge 0}\xi(f^n(x))>0.
\end{equation}
\end{lemma}

\begin{proof}
We use notation from the definition of $f$. Set
\[
A=\max(a_1,\dots,a_m),\ \ \ \alpha=\min(a_1,\dots,a_m).
\]
Since for every $i$ we have $0\le x_i\le 1$, we get
\begin{equation}\label{e2}
y_i\ge x_i\exp(-Ax_i)\ge x_i\exp(-A).
\end{equation}
There exists $k$ such that $x_k\ge 1/m$. Then
$y_k\le x_k\exp(-\alpha/m)$, so
\[
x_k-y_k\ge\frac1m\left(1-\exp\left(-\frac\alpha{m}\right)\right).
\]
Since $x_j\ge y_j$ for all $j$, we get
\[
1-Y\ge\frac1m\left(1-\exp\left(-\frac\alpha{m}\right)\right).
\]
Set
\[
C=\frac1m\left(m-1+\exp\left(-\frac\alpha{m}\right)\right).
\]
Then $Y\le C$ and $C<1$. By~\eqref{e2}, for every $i$ we get
\[
\frac{y_i}Y\ge\frac{x_i\exp(-Ax_i)}C\ge\frac{x_i\exp(-A)}C.
\]

We have
\[
\lim_{t\to 0}\frac{\exp(-At)}C=\frac1C>1,
\]
so there is $\varepsilon>0$ such that $x_i\le\varepsilon$ then
$y_i/Y\ge x_i$. Moreover, if $x_i\ge\varepsilon$ then
$y_i/Y\ge\varepsilon\exp(-A)/C$. This proves that~\eqref{e1} holds.
\end{proof}

We are ready to prove that the time average of the flow converges to the equilibrium flow.
 The update rule in the case of $m$ strategies is given by \eqref{mwu2}.

\begin{theorem}\label{t:Cesaro_multiple}
Given any non-atomic congestion game with $m$ actions as described by model (\ref{many_cost}),(\ref{mwu2}),
if $x=(x_1,\dots,x_m)\in\Delta$, and
$\min(x_1,\dots,x_m)>0$, then
\begin{equation}\label{e:Cesaro_multiple}
\lim_{T\to\infty}\frac1T\sum_{n=0}^{T-1} x_i(n)
=b_i.
\end{equation}
\noindent
where $(b_1, \dots, b_m )$ is the (Nash) equilibrium flow of the congestion game, i.e., $b_i= \frac{1/\alpha_i}{\sum_{j\{1,\dots,m\}}\alpha_j}$.
\end{theorem}

\begin{proof}
By dividing two equations of type (\ref{mwu2}) one for strategy $i$ and one for strategy $j$ we derive that:

\begin{equation}
\begin{aligned}
\frac{x_{n+1}(i)}{x_{n+1}(j)}&=\frac{x_{n}(i)}{x_{n}(j)} (1-\eps)^{c_n(i)-c_n(j)}\\
\end{aligned}
\end{equation}

By unrolling this relationship we derive:

\begin{equation}
\begin{aligned}
\frac{x_{n+1}(i)}{x_{n+1}(j)}&=\frac{x_{1}(i)}{x_{1}(j)}  (1-\eps)^{\sum_{\tau=1}^n \big(c_{\tau}(i)-c_{\tau}(j)\big)}\\
\end{aligned}
\end{equation}

By Lemma \ref{lem:interior}, given any initial condition  $(x_1(1), x_1(2), \dots, x_1(n))$ such that  $\min_i x_1(i)>0$ 
we have that there exists $\delta >0$  such that  $\inf_{n\ge 0}\min x_n(i)>\delta$ and $\sup_{n\ge 0}\max x_n(i)<1-\delta$. Then

\[ \frac{\delta}{1-\delta} \frac{x_{1}(j)}{x_{1}(i)}  <   (1-\eps)^{\sum_{\tau=1}^n \big(c_{\tau}(i)-c_{\tau}(j)\big)} <   \frac{1-\delta}{\delta} \frac{x_{1}(j)}{x_{1}(i)} \]
and
\[ \frac{1}{\ln(1-\eps)} \ln \big(  \frac{\delta}{1-\delta} \frac{x_{1}(j)}{x_{1}(i)} \big)  <   \sum_{\tau=1}^n \big(c_{\tau}(i)-c_{\tau}(j)\big) <    \frac{1}{\ln(1-\eps)} \ln \big(  \frac{1-\delta}{\delta} \frac{x_{1}(j)}{x_{1}(i)} \big). \]
Dividing all sides of the inequality by $n$ we get
\[ \frac{\frac{1}{\ln(1-\eps)} \ln \big(  \frac{\delta}{1-\delta} \frac{x_{1}(j)}{x_{1}(i)} \big)}{n}  <  \frac{ \sum_{\tau=1}^n \big(c_{\tau}(i)-c_{\tau}(j)\big)}{n} <   \frac{ \frac{1}{\ln(1-\eps)} \ln \big(  \frac{1-\delta}{\delta} \frac{x_{1}(j)}{x_{1}(i)} \big) }{n}.\]

By taking limits we have that for any $i, j$ $\lim_{n} \frac{ \sum_{\tau=1}^n \big(c_{\tau}(i)-c_{\tau}(j)\big)}{n} = 0$. 
For any subsequence such that the limits $\lim_{n} \frac{ \sum_{\tau=1}^n c_{\tau}(i)}{n}$ exist for all $i$
we have that: 

\begin{equation}
\label{eq:time-averages}
\lim_{n\to\infty} \frac{ \sum_{\tau=1}^n c_{\tau}(i)}{n} = \lim_{n\to\infty} \frac{ \sum_{\tau=1}^n \big(c_{\tau}(i)-c_{\tau}(j)\big) + \sum_{\tau=1}^n c_{\tau}(j)}{n}  = \lim_{n\to\infty} \frac{ \sum_{\tau=1}^n c_{\tau}(j)}{n}
\end{equation}

Since the cost functions are linear, i.e., $c_n(j) = a_jN x_n$, equation \eqref{eq:time-averages} implies that:

\begin{equation}
\label{eq:convergence-avg}
 a_i \lim_{n\to\infty}  \frac{ \sum_{\tau=1}^n x_{\tau}(i)}{n} = a_j \lim_{n\to\infty} \frac{ \sum_{\tau=1}^n x_{\tau}(j)}{n}
\end{equation}

\noindent
and the point $\big(\lim_{n}  \frac{ \sum_{\tau=1}^n x_{\tau}(1)}{n} ,\dots,  \lim_{n}  \frac{ \sum_{\tau=1}^n x_{\tau}(m)}{n}\big)$ is the unique equilibrium flow of the congestion game. Clearly, the same argument can be made for any other subsequence such that $\lim_{n} \frac{ \sum_{\tau=1}^n c_{\tau}(i)}{n}$ exists by possibly defining its own subsequence so that all the other limits also exist (which is always possible due to compactness).
By \eqref{eq:time-averages}, \eqref{eq:convergence-avg} the value must once again agree with the unique equilibrium of the game. Hence, for any $i$,  $\lim_{n} \frac{ \sum_{\tau=1}^n c_{\tau}(i)}{n}$ exists and 
$\big(\lim_{n}  \frac{ \sum_{\tau=1}^n x_{\tau}(1)}{n} ,\dots,  \lim_{n}  \frac{ \sum_{\tau=1}^n x_{\tau}(m)}{n}\big)$ is the unique equilibrium flow.
\end{proof}
}


{
\section{Extensions to congestion games with polynomial costs}
\label{s:nonlinear}

For simplicity, we will return to the case with exactly two actions/paths. 
As usual we will denote by $c(j)$ the cost of selecting the
strategy number $j$ (when $x$ fraction of the agents choose the first strategy).
We will focus on cost functions which are monomials with the same degree\footnote{This is convenient as it immediately implies that the Price of Anarchy is equal to $1$, since the potential is equal $\frac{1}{p+1}$ of the social cost function in these games and thus the equilibrium flow minimizes both the potential and the social cost.} 

\begin{equation}\label{poly_cost}
c_1(x)=\alpha N^p x^p \hspace{50pt} c_2(x)=\beta N^p x^p, \hspace{80pt} 
 p \in \mathbb{N}.
\end{equation}

The (Nash) equilibrium flow corresponds to the unique split $(x_b^*,1-x_b^*)$ such that 
$c_1(x_b^*)=c_2(x_b^*)$. 

\subsection{Multiplicative weights with polynomial costs}

Once again, applying the multiplicative weights updates rule we get formula~\eqref{mwu}. 
By substituting into~\eqref{mwu} the values of the  polynomial cost functions
from~\eqref{poly_cost} we get:
\begin{equation}\label{mwu_poly}
\begin{aligned}
x_{n+1}&=\frac{x_n(1-\eps)^{\alpha N^p x^p_n}}{x_n(1-\eps)^{\alpha N^p x^p_n}+
(1-x_n)(1-\eps)^{\beta N^p (1-x_n)^p}}\\
&=\frac{x_n}{x_n+(1-x_n)(1-\eps)^{N^p (\beta (1-x_n)^p - \alpha  x^p_n)}}.
\end{aligned}
\end{equation}

We introduce the new variables
\begin{equation}\label{varp}
a=(\alpha+\beta) N^p \ln\left(\frac1{1-\eps}\right),\ \ \ b=\frac{\beta}{\alpha+\beta}.
\end{equation}


Once again, we see that $b=1/2$ if and only if the two paths are totally symmetric (same cost function). In this case 
$x^*_b=1/2$ as well as the equilibrium flow splits the total demand equally in both paths.

We will thus study the dynamical systems generated by the one-dimensional map:
\begin{equation}\label{map_poly}
\begin{aligned}
f_{a,b}(x)&=\frac{x}{x+(1-x)\exp(aP_b(x))}.\\
\end{aligned}
\end{equation}

Clearly $f_{a,b}:[0,1]\rightarrow [0,1]$, where $0<b<1$, $a>0$, and
$P_b(x)=(1-b)x^n-b(1-x)^n.$
We have $P_b(0)=-b$, $P_b(1)=1-b$, and $P_b$ is strictly increasing.
Therefore, there exists the unique point $x_b^*\in(0,1)$ such that
$P_b(x_b^*)=0$. Observe that $f_{a,b}(x_b^*)=x_b^*$ is exactly the equilibrium flow.
Moreover 

\begin{equation}\label{monomderiv} f_{a,b}'(x)=\frac{(1-ax(1-x)P'_b(x))\exp (aP_b(x))}{(x+(1-x)\exp (aP_b(x)))^2}.\end{equation}

From \eqref{monomderiv} we have that
\[f_{a,b}'(0)=\exp (-aP_b(0))=\exp(ab) >1\;\;\text{and}\;\; f_{a,b}'(1)=\exp (aP_b(1))= \exp(a(1-b))>1.\]
Thus $0$ and $1$ are repelling. 
This fact implies (the proof of this fact is the same as of Lemma 3.1 from \cite{CFMP}) that there exists  an invariant attracting subset $I_{a,b}$ of the unit interval.

Although the time-average convergence of the flow does not necessarily converge to the equilibrium flow, one can still prove a theorem analogous to Theorem \ref{t:Cesaro} that reflects the \textit{time-average of the costs} of different paths. Informally,  although the traffic flows through each path evolve chaotically from day-to-day, if an outsider observer was to keep track of their time-average cost, all paths would appear to be experience similar delays. It is due to this inability to learn a preferred path that chaos (as we will argue next) is self-sustaining even under polynomial cost functions despite the application of learning/optimizing dynamics. 



\begin{theorem}
If cost functions $c_1$, $c_2$ are given by \eqref{poly_cost}, then \begin{equation} \label{cost_Cesaro} \lim_{n\to\infty}\frac 1n \sum_{k=0}^{n-1}\left(c_1(x_k)-c_2(x_k)\right)=0.\end{equation}
\end{theorem}

\begin{proof}
There is a closed interval $I_{a,b}\subset (0,1)$
which is invariant and attracting for $f_{a,b}$. Thus, there is $\delta\in(0,1)$ such
that $I_{a,b}\subset (\delta,1-\delta)$.

Fix $x=x_0\in[0,1]$ and use our notation $x_n=f_{a,b}^n(x_0)$. By induction we get

\begin{equation} \label{monomrec} x_n=\frac{x}{x+(1-x)\exp \left[a\sum_{k=0}^{n-1}(c_1(x_k)-c_2(x_k))\right]}.\end{equation}

Assume that $x=x_0\in I_{a,b}$. Since $\delta<x_n<1-\delta$, we have
\[
\frac{x}{1-\delta}<x+(1-x)\exp\left(a\sum_{k=0}^{n-1}(c_1(x_k)-c_2(x_k))\right)
<\frac{x}{\delta},
\]
so
\[
\delta^2<x\frac{\delta}{1-\delta}<(1-x)\exp\left(a\sum_{k=0}^{n-1}(c_1(x_k)-c_2(x_k))
\right)<x\frac{1-\delta}{\delta}<\frac1\delta.
\]
Therefore
\begin{equation}\label{est1}
\delta^2<\exp\left(a\sum_{k=0}^{n-1}(c_1(x_k)-c_2(x_k)) \right)<\frac1{\delta^2},
\end{equation}
so
\[
\left|a\sum_{k=0}^{n-1}(c_1(x_k)-c_2(x_k))\right|<2\log(1/\delta).
\]
This inequality can be rewritten as
\[
\left|\frac1n\sum_{k=0}^{n-1}(c_1(x_k)-c_2(x_k))\right|<\frac{2\log(1/\delta)}{an},
\]
and~\eqref{cost_Cesaro} follows.

If $x\in(0,1)\setminus I_{a,b}$, then by the definition of $I_{a,b}$
 there is $n_0$ such that $f_{a,b}^{n_0}(x)\in
I_{a,b}$, so~\eqref{cost_Cesaro} also holds.
\end{proof}

\subsection{Proof of the existence of Li-Yorke chaos}

\begin{theorem}
 For any $b\in(0,1/2)\cup(1/2,1)$   there exists $a_0$ such that if
$a>a_0$ then $f_{a,b}$ given by (\ref{map_poly}) has a periodic orbit of period $3$, and
therefore it has periodic orbits of all periods, positive topological entropy and is Li-Yorke chaotic.
\end{theorem}

\begin{proof}
Fix $b\in(0,1/2)$. It is enough to show that if $a$ is sufficiently
large, then there exist $x_0,x_1,x_2,x_3$ such that
$f_{a,b}(x_i)=x_{i+1}$ and $x_3<x_0<x_1$.

Our points $x_i$ will depend on $a$. We start by taking
\[
x_1=1-\frac1a\ \ \textrm{and}\ \ y=\frac{x_b^*}2.
\]
Note that $y$ does not depend on $a$ and $P_b(y)<0$. The inequality
$f_{a,b}(y)>x_1$ is equivalent to
\[
y>(a-1)(1-y)\exp(aP_b(y)),
\]
which holds for sufficiently large $a$. Moreover, for sufficiently
large $a$ we have $f_{a,b}(x_b^*)<x_1$. Therefore, for sufficiently
large $a$ there exists $x_0\in(y,x_b^*)$ such $f_{a,b}(x_0)=x_1$. In
particular, we have $x_0<x_1$.

Set
\[
b^*=\frac34-\frac{b}2.
\]
Since $b<1/2$, we have $b^*<1-b$, so if $a$ is sufficiently large,
then $P_b(x_1)>b^*$, and thus
\[ 
x_2=f_{a,b}(x_1)=\frac{x_1}{x_1+\frac1a\exp(aP_b(x_1))}\le
\frac a{\exp(ab^*)}=a\exp(-ab^*). 
\]
Since $P_b(x_2)\ge-b$, we get
\[\begin{split}
x_3=f_{a,b}(x_2)&=\frac{x_2}{x_2+(1-x_2)\exp(aP_b(x_2))}\le
\frac{x_2}{x_2+(1-x_2)\exp(-ab)}\\
&=\frac{x_2\exp(ab)}{x_2\exp(ab)+1-x_2}\le x_2\exp(ab)\le
a\exp(a(b-b^*)).
\end{split}\]
Since
\[
b-b^*=b-\frac34+\frac{b}2=\frac{3(b-\frac12)}2<0,
\]
we have $\lim_{a\to\infty}a\exp(a(b-b^*))=0$, and therefore if $a$ is
sufficiently large, then $x_3<y<x_0$. Hence, $f_{a,b}$ has a periodic
orbit of period 3.

We have $f_{a,1-b}(1-x)=1-f_{a,b}(x)$, so $f_{a,1-b}$ is conjugate to
$f_{a,b}$. Therefore the theorem holds also for $b\in(1/2,1)$.
\end{proof}

We did not use too many properties of $P_b$, so the theorem holds for
a larger class of those functions.

\begin{corollary}
Given any non-atomic congestion game with polynomial cost functions described by model (\ref{mwu_poly}), 
except for the symmetric case with $\alpha=\beta$, then there exists a total system demand $N_0$ such that for if $N\geq N_0$ the system has periodic orbits of all periods, positive topological entropy and is Li-Yorke chaotic.
\end{corollary}

}


\section{Extensions to congestion games with heterogeneous users}
\label{s:mixture}

This is the model for the case of heterogeneous population. 
We will start with the simplest possible case where there are only two subpopulations.
We will consider a two-strategy \emph{congestion game}  with two continuums of players/agents, where all of them use  \emph{multiplicative
  weights update}. Each of the players controls an infinitesimal small fraction of the flow.
 Out of the total  flow/demand $N$  of the first population has size $N\eta_1$ whereas the total flow of the second population is $N\eta_2$.
 A canonical example would be $\eta_1=\eta_2=0.5$.
  
  We will denote the fraction of the players of the first (resp. second) population using the first strategy at time $n$ as $x_n$  
  (resp. $y_n$). The second strategy is chosen by $1-x_n$ (resp. $1-y_n$) fraction of the players. Intuitively, this model captures how two large population of players/cars (e.g. taxis versus normal cars) chooses between two alternative, parallel paths for going from point $A$ to point $B$.
 If a large fraction of the players choose the same
strategy, this leads to congestion/traffic, and the cost increases. We will
assume that the cost is proportional to the \emph{load}. If we denote
by $c(j)$ the cost of the player playing the
strategy number $j$, and the coefficients of proportionality are
$\alpha,\beta$, then we get
\begin{equation}\label{cost_mix}
\begin{aligned}
c(1)&=\alpha N (\eta_1x + \eta_2y), & c(2)&=\beta N (\eta_1+\eta_2 - \eta_1x - \eta_2y)
\end{aligned}
\end{equation}

For \emph{multiplicative weights update} (MWU), for the first (resp. second), there is a
parameter $\eps_1\in(0,1)$, (resp. $\eps_2\in(0,1)$) which can be treated as the  common learning rate of
all players of that population. Thus, we get

\begin{equation}\label{mwu_mix}
\begin{aligned}
x_{n+1}&=\frac{x_n(1-\eps_1)^{c(1)}}{x_n(1-\eps_1)^{c(1)}+
(1-x_n)(1-\eps_1)^{c(2)}},\\
y_{n+1}&=\frac{y_n(1-\eps_2)^{c(1)}}{y_n(1-\eps_2)^{c(1)}+
(1-y_n)(1-\eps_2)^{c(2)}}.
\end{aligned}
\end{equation}


By combining equations $(\ref{cost_mix})$ and $(\ref{mwu_mix})$, 

\begin{equation}\label{model_mix}
\begin{aligned}
x_{n+1}&=\frac{x_n}{x_n+
(1-x_n)(1-\eps_1)^{N(\beta (\eta_1+\eta_2) - (\alpha+\beta)(\eta_1x_n+\eta_2y_n))}},\\
y_{n+1}&=\frac{y_n}{y_n+
(1-y_n)(1-\eps_2)^{N(\beta (\eta_1+\eta_2) - (\alpha+\beta)(\eta_1x_n+\eta_2y_n)) }}.
\end{aligned}
\end{equation}

After a similar change of variables as in the homogeneous case
 formula~\eqref{model_mix} becomes
\begin{equation}\label{mix1}
\begin{aligned}
x_{n+1}&=\frac{x_n}{x_n+(1-x_n)\exp(a_1(\eta_1x_n+\eta_2y_n-b))},\\
y_{n+1}&=\frac{y_n}{y_n+(1-y_n)\exp(a_2(\eta_1x_n+\eta_2y_n-b))}.\\
\end{aligned}
\end{equation}

In the simplest case of the equal shares/mixtures (i.e. $\eta_1=\eta_2=0.5$) we have:

\begin{equation}
\begin{aligned}\label{map_mix}
x_{n+1}&=\frac{x_n}{x_n+(1-x_n)\exp(a_1(0.5(x_n+y_n)-b))},\\
y_{n+1}&=\frac{y_n}{y_n+(1-y_n)\exp(a_2(0.5(x_n+y_n)-b))}.\\
\end{aligned}
\end{equation}

{\bf Dimensional reduction}:
Although the heterogeneous model contains more independent variables than the homogeneous case, the dynamics are constrained in a lower-dimensional manifold. That is, we will show that the function  $I(x,y)=\frac {(1-x)^{a_2}y^{a_1}} {(1-y)^{a_1}x^{a_2}}$ is an invariant function for population mixtures.
This means that the curves $I(x,y)=c$ are invariant for any time step $n$, where $c$ parametrizes the family of invariant curves.
\begin{lemma}
The function  $I(x,y)=\frac {(1-x)^{a_2}y^{a_1}} {(1-y)^{a_1}x^{a_2}}$ is an invariant function (first integral) of the dynamics.
\end{lemma}

\begin{proof}
It is easy to check that  the set of equations (\ref{mix1}) is equivalent to

\begin{equation}
\begin{aligned}
\frac{x_{n+1}}{1-x_{n+1}}&=\frac{x_n}{(1-x_n)\exp(a_1(\eta_1x_n+\eta_2y_n-b))},\\
\frac{y_{n+1}}{1-y_{n+1}}&=\frac{y_n}{(1-y_n)\exp(a_2(\eta_1x_n+\eta_2y_n-b))}.\\
\end{aligned}
\end{equation}

By raising the first equation to power $a_2$ and the second equation to power $a_1$ and dividing them we derive that:

$$\frac{x_{n+1}^{a_2}(1-y_{n+1})^{a_1}}{(1-x_{n+1})^{a_2}y_{n+1}^{a_1}} = \frac{x_{n}^{a_2}(1-y_{n})^{a_1}}{(1-x_{n})^{a_2}y_{n}^{a_1}}$$

That is the function  $I(x,y)=\frac {(1-x)^{a_2}y^{a_1}} {(1-y)^{a_1}x^{a_2}}$ is an invariant function (first integral) of the dynamics.
\end{proof}

{\bf Time-average convergence of the mixture to Nash equilibrium $b$.}
For the considered heterogeneous model we can show a result similar to Theorem \ref{t:Cesaro} for the homogeneous population, that is that  $b$ is  Ces\`{a}ro attracting mixture of trajectories. 

\begin{theorem}\label{t:Cesaro-mix}
For every $a_1,a_2>0$, $b\in(0,1)$ and $(x_0,y_0) \in(0,1)^2$ we have
\begin{equation}\label{e:Cesaro_mixture}
\lim_{T\to\infty}\frac1T\sum_{n=0}^{T-1}\left(\eta_1x_n+\eta_2 y_n\right)=b.
\end{equation}
\end{theorem}

\begin{proof}
Let $f(x_n,y_n)=(x_{n+1},y_{n+1})$ be defined by (\ref{mix1}) where $\eta_1,\eta_2\in(0,1)$ and $\eta_1+\eta_2=1$.

The map $t\mapsto\frac t{1-t}$ is a homeomorphism of $(0,1)$ onto
$(0,\infty)$, and its inverse is given by $t\mapsto\frac t{1+t}$.
Thus, we can introduce new variables, $z=\frac x{1-x}$ and $w=\frac
y{1-y}$. In these variables our map will be $g:(0,\infty)^2\to
(0,\infty)^2$, and if $g(z_n,w_n)=(z_{n+1},w_{n+1})$, then
\begin{equation}\label{two}
\begin{aligned}
z_{n+1}&=z_n\exp\left(-a_1\left(\eta_1\frac{z_n}{1+z_n}+\eta_2
\frac{w_n}{1+w_n}-b\right)\right),\\
w_{n+1}&=w_n\exp\left(-a_2\left(\eta_1\frac{z_n}{1+z_n}+\eta_2
\frac{w_n}{1+w_n}-b\right)\right).
\end{aligned}
\end{equation}

If $w_n=cz_n^{a_2/a_1}$ then $w_{n+1}=cz_{n+1}^{a_2/a_1}$. This shows
that if $z_n$ is close to 0 then also $w_n$ is close to 0, and
by~\eqref{two} we get $z_{n+1}>z_n$. Similarly, if $z_n$ is close to
infinity, then also $w_n$ is close to infinity, and by~\eqref{two} we
get $z_{n+1}<z_n$. Together with another inequality obtained
from~\eqref{two},
\[
z_n\exp(-a_1(1-b))<z_{n+1}<z_n\exp(a_1b),
\]
this proves that if $z_0,w_0\in(0,\infty)$ then $\inf_{n\ge 0}z_n>0$
and $\sup_{n\ge 0}z_n<\infty$.

The first equation of~\eqref{two} can be rewritten as
\[
z_{n+1}=z_n\exp(a_1(\eta_1x_n+\eta_2y_n-b)),
\]
so by induction we get
\[
z_T=z_0\exp\left(a_1\left(\sum_{n=0}^{T-1}(\eta_1x_n+\eta_2y_n)-Tb
\right)\right).
\]
Therefore there exists a real constant $M$ (depending on the
parameters and the initial point $(x_0,y_0)$), such that
\[
\left|\sum_{n=0}^{T-1}(\eta_1x_n+\eta_2y_n)-Tb\right|\le M
\]
for every $T$. Dividing by $T$ and passing to the limit, we get
\begin{equation}
\lim_{T\to\infty}\frac1T\sum_{n=0}^{T-1}\left(\eta_1x_n+\eta_2y_n\right)=b.
\end{equation}
\end{proof}

We end here with numerical results to demonstrate that, perhaps not surprisingly, this class of games not only possesses complex non-equilibrium behavior, but also allows for an immediate generalization to a more realistic, larger dimensional system, in which new and even more complex non-equilibrium phenomena can arise. Developing a more complete theoretical understanding of these issues, will likely require the introduction of new tools and techniques.

Figures (\ref{fig:hetero1}) and (\ref{fig:hetero2}) show attracting orbits generated from the map (\ref{map_mix}) (with $\eta_1=\eta_2=0.5$) for fixed values of $a_1, a_2, b$. There, $5000$ random starting points are initialized. To approximate where the attractors lie, the first $1000$ iterates were made without plotting; the next $200$ were visualized.

\begin{figure}
  \centering
  \includegraphics[width=.5\linewidth]{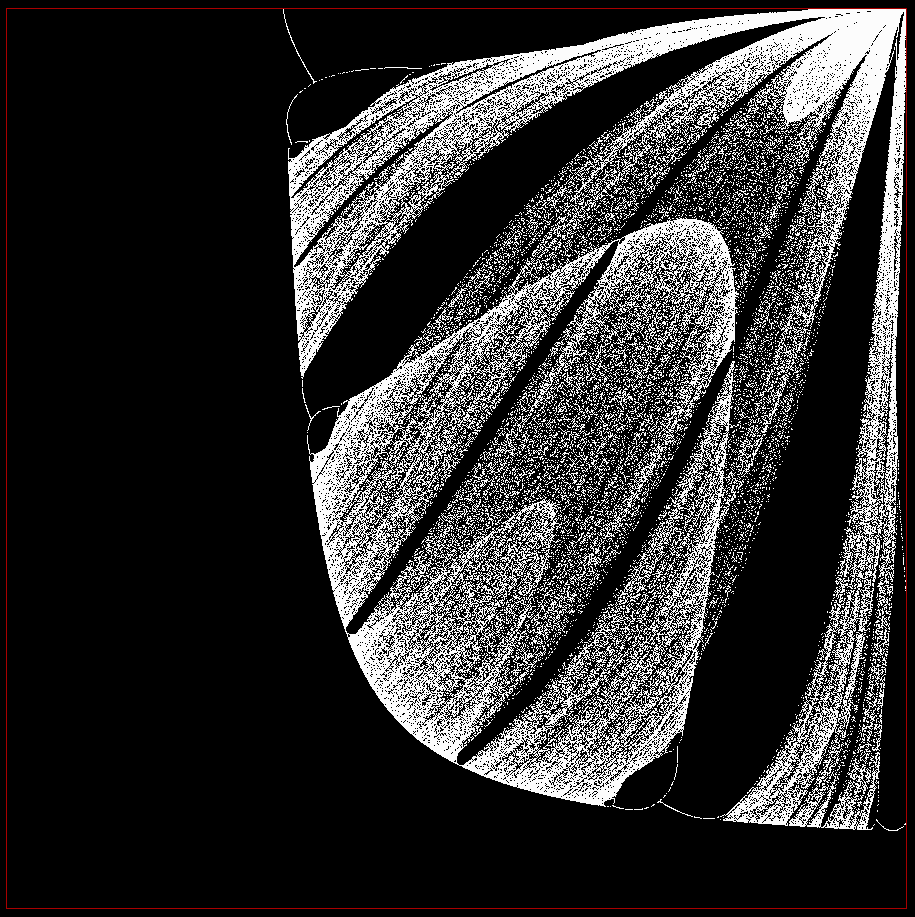}   
  \caption{Attractor for the map (\ref{map_mix}) in the two-subpopulation model with $a_1=20, a_2=30, b=0.8$. The white dots are the coordinates $(x,y)$ generated from initializing 5000 $(x_0,y_0)$'s at random from the unit square domain, iterating them with (\ref{map_mix}) 1000 times, then visualizing the next 200 iterates. 
}
  \label{fig:hetero1}
\end{figure}%

\begin{figure}
  \centering
  \includegraphics[width=.5\linewidth]{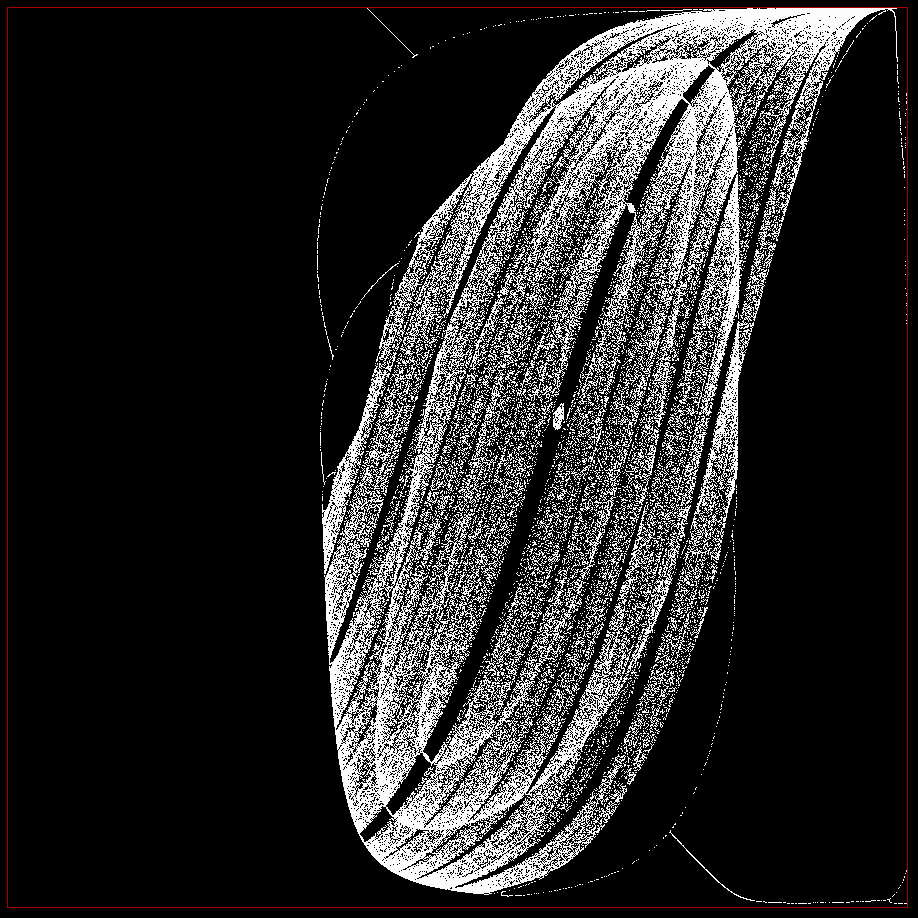}  
  \caption{Attractor for the map (\ref{map_mix}) in the two-subpopulation model with $a_1=10, a_2=30, b=0.7$. The white dots are the coordinates $(x,y)$ generated from initializing 5000 $(x_0,y_0)$'s at random from the unit square domain, iterating them with (\ref{map_mix}) 1000 times, then visualizing the next 200 iterates.  
  }
  \label{fig:hetero2}
\end{figure}

\section{Chaos in large atomic congestion games via reductions to the non-atomic case} 
\label{s:reductions}

In this paper we have focused on analyzing MWU in (mostly linear) non-atomic congestion games. In these settings each individual agent is assumed to control an infinitesimal account of the overall flow $N$. In the atomic setting each agent is controlling a discrete unsplittable amount of flow, i.e., a packet of size $1$. There are now $N$ agents that need to choose amongst the different paths that are available to them. In this section, we will show how to translate results from the case of non-atomic congestion games to their atomic counterparts. To do so we will show that the MWU maps in the case of linear atomic congestion games can be reduced to MWU maps of non-atomic games, which we have already analyzed.

 We will study MWU in a linear congestion game under the easiest information theoretic model of full information where on every day MWU receives as an input the expected cost of all actions/paths. Furthermore, we will assume that all $N$ agents are initialized with the fully mixed (interior) strategy $x$. This of course is not a generic initial condition but since we are working towards negative/complexity/chaos type of results we can choose our initial condition in an adversarial manner. Due to the symmetry of initial conditions, the payoff vectors that agents experience  on any day are common across all agents. Hence, the symmetry of initial conditions is preserved. For such trajectories we only need to keep track of a single probability distribution (the same one for all agents), which is already reminiscent of the non-atomic setting where we only have to keep track of the ratios/split of the total demand along the different paths/strategies. Let's denote by $x$ this common probability vector for all $N$ agents. We are ready to define our model in detail. 

{\bf Atomic model with $N$ agents/players.}  We have $N$ agents. Each agent can choose between $m$ strategies/paths.
The cost function for each strategy/path is a linear function on the number of the agents using that path $i$, ie, a linear function of its load. Let $\alpha_i$ be the respective multiplicative constant for strategy $i$. Suppose all agents use the same probability distribution $x$. The expected cost of any agent for using strategy $i$ is

\begin{equation}\label{many_cost_discrete}
\begin{aligned}
c(i)&=\alpha_i (1 + (N-1) x_i) ~\forall  i \in \{1, \dots, m\}.
\end{aligned}
\end{equation}

Given this payoff vector, the MWU updates follow the same format as always.
At time $n+1$ the players know already the  expected cost of the strategies at
time $n$ and update
their choices. 
 The update rule in the case of $m$ strategies is as follows:

\begin{equation}\label{mwu_discrete}
\begin{aligned}
x_i(n+1)&=x_i(n)\frac{(1-\eps)^{c(i)}}{\sum_{j \in \{1,\dots,m 
\}} x_j(n)(1-\eps)^{c(j)}},\\
\end{aligned}
\end{equation}

We are now ready to state two formal results. One for the case of games with two strategies and one for the more general case with $m$ strategies. 

\begin{theorem}
\label{t:Disc1}
Let's consider an atomic congestion game with $N$ agents and two paths of linear cost functions as described by equations (\ref{many_cost_discrete}), (\ref{mwu_discrete}). 
Let $x$ be an interior probability distribution that is a common initial condition for all $N$ agents. As long as the congestion game has a symmetric interior Nash equilibrium where both agents play the distribution $(p, 1-p)$ with $0<p<1$\footnote{The game has a symmetric interior Nash if and only if   $\alpha_2 <  N \alpha_1$ and 
$\alpha_1 <  N \alpha_2$. 
} the update rule of the probability distribution $x$ under MWU dynamics is as in the case of the non-atomic model map (\ref{map}) where  $a= (N-1) (\alpha_1+\alpha_2) \ln \big( \frac{1}{1-\epsilon} \big)$ and $b=p$. Thus, as long as $p\neq 0.5$,  there exists a threshold capacity $N_0$ such that if  the number of agents $N$ exceeds $N_0$ the system has periodic orbits of all possible periods, positive topological entropy and is Li-Yorke chaotic. If $p=0.5$,  although the Price of Anarchy of the game converges to one as $N\rightarrow \infty$, the time-average social cost can be arbitrarily close to its worst possible value.
\end{theorem}

\begin{proof}
By substituting into~\eqref{mwu_discrete} the values of the cost functions
from~\eqref{many_cost_discrete} we get:
\begin{equation} 
\begin{aligned}
x_{n+1}&=\frac{x_n(1-\eps)^{\alpha_1 (1+ (N-1) x_n)}}{x_n(1-\eps)^{\alpha_1 (1+ (N-1) x_n)}+
(1-x_n)(1-\eps)^{\alpha_2 (1+ (N-1) (1-x_n))}}\\
&=\frac{x_n}{x_n+(1-x_n)(1-\eps)^{\alpha_2 N-\alpha_1 -(\alpha_1+\alpha_2) (N-1) x_n}}.
\end{aligned}
\end{equation}

We introduce the new variables
\begin{equation}
a=(N-1) (\alpha_1+\alpha_2) \ln \big( \frac{1}{1-\epsilon} \big),\ \ \ b=\frac{\alpha_2 N-\alpha_1}{(\alpha_1+\alpha_2) (N-1)}.
\end{equation}

Note that the symmetric strategy where all agents play according to $(b,1-b)$ is an interior Nash equilibrium. Given this new formulation we see that the map is the same as the one for the non-atomic case (\ref{map}). The claims about chaos follow by direct application of Corollary \ref{chaos}. In the case where $p=0.5$, we have that the uniform distribution is an interior Nash and this implies that $\alpha_1=\alpha_2(=\alpha)$, i.e. both paths have the same cost function. In terms of Price of Anarchy, the expected cost of any agent at a Nash equilbrium is at most $\alpha(1+\frac{N-1}{2})=\alpha \frac{N+1}{2}$. Hence the social cost of any Nash equilibrium is at most $\alpha \frac{N(N+1)}{2}$. On the other hand, the socially optimal state that divides the load as equally as possible has cost at least $\alpha \frac{N^2}{2}$ and the ratio of the two converges to $1$ as $N\rightarrow \infty$. Finally, the fact that there exist trajectories such that the time-average social cost can be arbitrarily close to its worst possible value follows from a direct application of Theorem \ref{t:SC} given the equivalence of the update rule for the atomic and non-atomic case and the fact for both systems (approximate) worst case performance is experienced when (in expectation almost) 
all users/flow are using the same strategy.
\end{proof}

Theorem \ref{t:Disc1} applies for atomic congestion games with numerous agents but only two paths. As we show next, chaos is robust and emerges in atomic congestion games regardless of the number of available paths.

\begin{theorem}
\label{t:Disc2}
Let's consider an atomic congestion game with $N$ agents and $m$ paths of linear cost functions as described by equations (\ref{many_cost_discrete}), (\ref{mwu_discrete}). Let the cost functions of the all paths be $\alpha x$ where $x$ the load of the respective path and $\alpha$ the common multiplicative constant. Let $x$ be interior probability distribution that is a common initial condition for all $N$ agents.  The update rule of the probability distribution $x$ under MWU dynamics is as in the case of the non-atomic model map (\ref{f1m}) where $a_i= (N-1) \alpha \ln \big (  \frac{1}{1-\epsilon} \big)$. Thus, for any such atomic congestion game there exists a threshold capacity $N_0$ such that if  the number of agents $N$ exceeds $N_0$ the system has periodic orbits of all possible periods, positive topological entropy and is Li-Yorke chaotic.
\end{theorem}

\begin{proof}
The reduction of the map described by equations (\ref{many_cost_discrete}), (\ref{mwu_discrete}) to non-atomic model map (\ref{f1m}) follows easily once we observe that MWU, i.e. map (\ref{mwu_discrete}) is invariant to shifts of the cost vector by a constant value.\footnote{This invariance is also true for most standard regret minimizing dynamics, e.g. Follow-the-Regularized-Leader.} That is, for any $\gamma$ if we apply the vector $c'(i)= c(i)+\gamma$ to map (\ref{mwu_discrete}) it remains unchanged. Hence instead of substituting into~\eqref{mwu_discrete} the values of the cost functions $c(i)=\alpha (1+(N-1)x)$, we instead substitute the values  $c'(i)=\alpha (N-1)x$. However, this is exactly map in the case of the non-atomic model map (\ref{f1m}) with $a_i= (N-1) \alpha \ln \big (  \frac{1}{1-\epsilon} \big)$. The rest of the theorem follows immediately by applying Theorem \ref{t:chaos_many}.
\end{proof}


\end{document}